\documentclass[12pt]{article}

	\usepackage{amsmath, graphicx, fancyhdr, tikz}
	\usepackage[margin=1in]{geometry}
	\usepackage[ansinew]{inputenc} 
	\usepackage{epsfig}
	\usepackage{color}
	\usepackage{amssymb}
	\usepackage{amsthm}
	\usepackage[sectionbib]{natbib}
	\usepackage{pdfpages}
	\usepackage{changes}
	\usepackage{float}

	\usepackage{lineno}

	\usepackage{titlesec}
	\titleformat{\section}{}{}{0em}{\Large\scshape\thesection.~}
	\titleformat{\subsection}{}{}{0em}{\large\scshape\thesubsection.~}
	\titleformat{\subsubsection}{}{}{0em}{\scshape\thesubsubsection.~}
\usepackage{authblk}

\usepackage{setspace}
\usepackage[pdfpagelabels]{hyperref}
	\newtheoremstyle{mystyle}
  {}
  {}
  {\itshape}
  {}
  {\normalsize \scshape}
  {.}
  { }
  {\thmname{#1}\thmnumber{ #2}\thmnote{ (#3)}}

  \theoremstyle{mystyle}
	
		\newcommand{\E}{\mathbb{E}}
		
		\newcommand{\Q}{\mathbb{Q}}
		\renewcommand{\P}{\mathbb{P}}
		\newcommand{\R}{\mathbb{R}}
		\newcommand{\D}{\mathcal{D}}

		\newcommand{\X}{\mathbf{X}}
		\newcommand{\Y}{\mathbf{Y}}

		\newcommand{\EE}{\mathcal{E}}
		\newcommand{\U}{\mathcal{U}}
		\newcommand{\N}{\mathcal{N}}

		\newcommand{\brae}[1]{\left[ #1 \right]}
		\newcommand{\bra}[1]{\left( #1 \right)}
		
		\newcommand{\norm}[1]{\left\| #1 \right\| }
		\newcommand{\abs}[1]{\left| #1 \right| }
		\newcommand{\eps}{\varepsilon}
		\newcommand{\titname}[1]{titlename}

		\newtheorem{theorem}{Theorem}[section]
		
		\newtheorem{definition}[theorem]{Definition}
		\newtheorem{notation}[theorem]{Notation}
		\newtheorem{example}[theorem]{Example}

		\newtheorem{remark}[theorem]{Remark}
\usetikzlibrary{calc}
\usetikzlibrary{arrows}
\usetikzlibrary{decorations.markings}
\usetikzlibrary{through}
\usetikzlibrary{patterns,snakes}

\tikzset{
    right angle quadrant/.code={
        \pgfmathsetmacro\quadranta{{1,1,-1,-1}[#1-1]}     
        \pgfmathsetmacro\quadrantb{{1,-1,-1,1}[#1-1]}},
    right angle quadrant=1, 
    right angle length/.code={\def\rightanglelength{#1}},   
    right angle length=2ex, 
    right angle symbol/.style n args={3}{
        insert path={
            let \p0 = ($(#1)!(#3)!(#2)$) in     
                let \p1 = ($(\p0)!\quadranta*\rightanglelength!(#3)$), 
                \p2 = ($(\p0)!\quadrantb*\rightanglelength!(#2)$) in 
                let \p3 = ($(\p1)+(\p2)-(\p0)$) in  
            (\p1) -- (\p3) -- (\p2)
        }
    }
}

\tikzset{%
  add/.style args={#1 and #2}{to path={%
 ($(\tikztostart)!-#1!(\tikztotarget)$)--($(\tikztotarget)!-#2!(\tikztostart)$)%
  \tikztonodes}}
}

	\AtBeginDocument{\addtocontents{toc}{\protect\setlength{\parskip}{0pt}}}
	\renewcommand{\dagger}{**}
	\renewcommand{\ddagger}{{**}*}
	\graphicspath{{./figures/}}

    \makeatletter
    \renewcommand*{\@fnsymbol}[1]{\ensuremath{\ifcase#1\or *\or \dagger\or \ddagger\or
       \mathsection\or \mathparagraph\or \|\or **\or \dagger\dagger
       \or \ddagger\ddagger \else\@ctrerr\fi}}
    \makeatother

\title{Noise Fit, Estimation Error and a\\Sharpe Information Criterion}
\date{December 11, 2019}

\author[1]{Dirk Paulsen\thanks{John Street Capital LLP, London, United Kingdom, Email:
dirk.h.paulsen@googlemail.com}}
\author[2]{Jakob S\"ohl\thanks{Delft Institute of Applied Mathematics, TU Delft, The Netherlands, Email: j.soehl@tudelft.nl}}
\affil[1]{John Street Capital}
\affil[2]{TU Delft}

\begin{document}
\maketitle

\pagenumbering{arabic}
\begin{abstract}
When the in-sample Sharpe ratio is obtained by optimizing over a $k$-dimensional parameter space, it is a biased estimator for what can be expected on unseen data (out-of-sample). We derive (1) an unbiased estimator adjusting for both sources of bias: noise fit and estimation error. We then show (2) how to use the adjusted Sharpe ratio as model selection criterion analogously to the Akaike Information Criterion (AIC). Selecting a model with the highest adjusted Sharpe ratio selects the model with the highest estimated out-of-sample Sharpe ratio in the same way as selection by AIC does for the log-likelihood as measure of fit.
\end{abstract}

\newpage

	\section{Introduction}

A convenient measure  for return predictability is the \textit{Sharpe ratio}. The Sharpe ratio is the ratio of annualized mean (excess) returns over return volatility. It is a useful statistic as it summarizes the first two moments of the return distribution, is invariant under leverage, and is therefore, to a second order approximation, the relevant metric regardless of the investors risk-return preferences.

In many situations, the Sharpe ratio is maximized over a set of parameters on an in-sample data set. Be it weights in a portfolio of assets (\cite{markowitz1952portfolio}), exposure to (risk) factors, predictor variables in a regression or other parameters influencing an investment strategy like the time-horizon of a trend model.

The thus optimized  (in-sample) Sharpe ratio overestimates the Sharpe ratio that can be  expected on unseen data (out-of-sample). The reason is twofold. First, the in-sample Sharpe ratio makes use of the in-sample data set twice. A first time to estimate the optimal parameter and a second time to estimate the resulting Sharpe ratio. As the estimated parameter will be tuned towards the noise of the in-sample data, calculating the in-sample Sharpe ratio on the same data set will overestimate the  Sharpe ratio of the true parameter (noise fit). Second,   the estimated parameter deviates from the true parameter. Thus the Sharpe ratio at the estimated parameter will be smaller than at the true parameter on unseen data (estimation error). For a more formal definition of noise fit and estimation error we refer to the error decomposition in equation (\ref{eq_decomposition}).
 
As illustration consider a linear regression model. Suppose we are interested in the question to what degree returns of the stocks in the S\&P500 Index are predictable by characteristics like the price-to-earnings ratio and the dividend yield. In this case we would model the return $r_t^i$ of stock $i=1,\ldots,500$ as
$$r_t^i = \mu_t^i + \eps_t^i, \quad \text{with } \mu_t^i = \sum\limits_{j=1}^M x_t^{i,j}\theta_j$$
 where $\eps_t^i$ is noise of zero mean and the expected stock return $\mu_t^i$ is linear in the $M$ characteristics $x_t^{i,j}, j=1,\ldots,M$ with unknown coefficients $\theta$. We can get an estimate $\hat \theta$ by generalized least squares or, yielding the same, by maximizing the Sharpe ratio of a portfolio using the appropriately scaled return predictions as weights.  
If we estimate the parameters on the in-sample data set and measure predictability on the same data, say by $R^2$ or the Sharpe ratio, we overestimate the true predictability (noise fit). To wit, even if no predictability exists, noise fit will guarantee that we will still observe a positive in-sample $R^2$.

Second, if we use the estimated coefficients $\hat \theta$ and apply them on an independent data set, we underestimate the true predictability as estimation error in the estimated coefficients diminishes their predictive power. The estimated model has lower predictive power than
the unknown true model (estimation error).

Standard statistical tests deal with inference for true models. Standard tests ask whether $\hat \theta$ allows to conclude that the predictability of the unknown true $\theta^*$ is larger than a threshold at a given confidence level.  They thus correct in-sample estimates for noise fit\footnote{For instance, in the F-test for a linear regression this correction is implicit in the non-zero mean of the F-distribution by which the F-test takes into account that there will be positive in-sample fit even under the null hypothesis of zero predictability.}. In this paper, by contrast, we are interested in the Sharpe ratio at the estimated parameters $\hat\theta$. We therefore need to additionally correct for estimation error.  

What Sharpe ratio can a portfolio manager expect out-of-sample when applying the estimated model? What is the best estimate of the out-of-sample Sharpe ratio of a fitted model? 

We answer that question by deriving an unbiased closed form  estimator for the out-of-sample Sharpe ratio when the in-sample Sharpe ratio is obtained by fitting $k$ parameters. Our estimator corrects for both, noise fit and estimation error, and can be calculated based on observable data only. In particular, we do not assume the true Sharpe ratio to be known. This is also why the estimator can be applied as model selection criterion, our second contribution. The resulting criterion is analogous to the Akaike Information Criterion but with the Sharpe ratio as objective rather than log-likelihood (see Section \ref{section_model_selection}). We therefore call our estimator \textit{SRIC} for \textit{Sharpe Ratio Information Criterion}. 

The detrimental effect of estimation error on out-of-sample performance in portfolio selection (or return prediction for that matter) has long been recognized (see e.g. \cite{frankfurter1971portfolio,dickinson1974reliability,jobson1980estimation}). \cite{demiguel2009optimal} suggest estimation error is so severe that simple allocation rules that refrain from optimization at all are often superior to naive optimization. 
 
A variety of ideas have emerged to combat estimation error. \cite{jorion1986bayes}, among others, suggests to shrink return estimates. \cite{kan2007optimal} propose rules to combine the Markowitz portfolio with the riskless asset and the minimum variance portfolio in order to maximize an estimate of out-of-sample mean variance utility. 
\cite{kirby2012optimizing} suggest to parametrize investment strategies. While this does not eliminate estimation error completely, the problem is reduced to a potentially lower dimensional subspace and this allows to use a longer estimation horizon. To wit, the optimal weighting between an estimate of the mean-variance optimal portfolio and the minimum variance portfolio might be well estimated over $50$ years of data while the return of a particular stock might be best estimated by looking back no further than $1$ year. The parametrization approach is akin to the regression setting that is covered by this paper.  

\cite{ledoit2014nonlinear} answer the question how to shrink the covariance matrix in order to optimize the expected out-of-sample Sharpe ratio. The mean return is assumed to be known but the covariance matrix has to be estimated, which is the opposite to our set-up. 

For ordinary least square regressions the adjusted $R^2$ corrects the in-sample $R^2$ for noise-fit but not estimation error\footnote{ This can be seen as under the true model the expected value of the adjusted $R^2$ equals the true $R^2$ (assuming the true variance of the dependent variable is known and equal to the unbiased sample estimate).}. In more general settings (e.g. regression or also the setting at hand, see below), the Akaike Information Criterion (AIC) - see \citet*[around equation (2)]{akaike1974new}, \cite{akaike1998bayesian}, \cite{akaike1998information} - asymptotically adjusts for both, noise fit and estimation error, if performance is measured in terms of the log-likelihood function\footnote{Though the AIC is formulated in terms of information theory - the AIC minimizes the estimated Kullback-Leibler divergence between the selected and the true model - there is a different, in many contexts more intuitive, interpretation. By construction AIC is an estimate of the Kullback-Leibler divergence between the true and estimated model. The Kullback-Leibler divergence can be written as the difference of the expected out-of-sample log-likelihood between the true model and the fitted model. Hence choosing the minimum AIC model chooses the model with the highest estimated out-of-sample log-likelihood. 
See also \cite{stone1977asymptotic} who shows that in a regression context minimizing the AIC is asymptotically equivalent to minimizing the cross validation error (sum of squared residuals) or 
\citet*[page 61]{burnham2002model} or Section \ref{section_relation_to_aic} which proves the statement for the setting at hand. \nopagebreak
}. The model that minimizes AIC is the model with the highest estimated out-of-sample log-likelihood. Here, by contrast, we are concerned with the problem of maximizing the Sharpe ratio rather than log-likelihood.

For more general objective functions, \cite{west1996asymptotic} provides tools to calculate moments of smooth functions of out-of-sample predictions. The suggested technique is to use a Taylor expansion around the true parameters. This is not necessary in the linear case at hand. \cite{hansen2009sample} derives the joint asymptotic distributions for noise fit and estimation error in general settings as a function of other limit entities. 

In spite of this literature, there is surprisingly little research on the quantification of estimation error and noise fit or on how many parameters to maximize upon. When predicting returns one wants to know whether the prediction is good enough to promise sufficiently high out-of-sample performance or when including an additional predictor is detrimental to it.

Most similar and at the same time complementary to what we aim at in this article is \cite{siegel2007performance}. The authors derive an asymptotically unbiased estimator for the mean out-of-sample return and the out-of-sample variance of mean-variance portfolios. This is very useful when estimating the efficient frontier. For the Sharpe ratio, however, the resulting estimator is biased of order $\frac{1}{T}$, the same order as without bias adjustment (see our appendix for a derivation). In addition, while technically more involved, the authors consider the case in which the weights sum to $1$. We, by contrast, do not make this constraint as first, we want to allow for leverage and second, we are mainly interested in the case in which the parameters are not weights but regression coefficients. 

\cite{karoui2013realized} estimates the out-of-sample variance when the number of assets is large. 

For the Sharpe ratio,  \cite{bailey2014deflated} and \cite{harvey2015backtesting}  derive estimates of the out-of-sample Sharpe if the in-sample Sharpe ratio is obtained by maximizing over $N$ trials. Their approach is based on computing $p$-values for multiple hypothesis testing and using them as a heuristic for the Sharpe ratio. \cite{novymarx2015} looks at critical values for selecting the $k$ best out of $N$ independent signals. While critical values for hypothesis testing essentially correct for the noise fit in the data, they do not adjust for the estimation error.\footnote{This is as statistical standard tests are designed for inference about the predictive power of the unknown \textit{true parameter} but not about the parameter that will be applied out-of-sample, namely the  noise contaminated \textit{estimated parameter}. This can be best seen for the likelihood ratio test, whose test statistic (in the Gaussian case) is $\frac{1}{2}\chi^2(k)$ distributed, while the difference between in-sample and out-of-sample log likelihood is $\frac{1}{2}\chi^2(2 k)$ distributed which is why the Akaike Information Criterion punishes the number of parameters by a factor $2$ (see also Section 4.2). }
 
Recently, \cite{kourtis2016sharpe} derived an approximate correction for the estimation error (but not noise fit) for the squared Sharpe ratio obtained from maximizing over a k-dimensional parameter space under the assumption that the true Sharpe ratio is known. As the squared Sharpe ratio (which can turn negative numbers positive) approximates mean-variance utility, this is up to order $o(1/T)$ the same correction as in \citet*[equation (16)]{kan2007optimal}, who derive an estimator for the out-of-sample mean variance utility in dependence of the unknown optimal one. Their adjustment is half the adjustment in the AIC as there is no correction for noise fit. 

\cite{kan2015economic}, among other results, calculate the distribution of the out-of-sample Sharpe ratio. However, also their result depends on knowledge of the true Sharpe ratio. In contrast, we do not assume the true Sharpe ratio to be known and provide an estimator based purely on observable data.





Two applications of SRIC lie on the hand. First, an investor maximizing an investment strategy over $k$ parameters, that could be assets in her portfolio or return predictors, might be interested in an estimate of its out-of-sample Sharpe ratio. 

Second, an investor might be interested in how many parameters to maximize her strategy upon. Do price-dividend ratios add to performance if one already has a factor based on price-earning ratios? This is a question of model selection.  The Akaike Information Criterion (AIC) selects the portfolio which optimizes an estimate of the out-of-sample log-likelihood. In the Gaussian case this is the same as maximizing an estimate of out-of-sample mean-variance utility.  SRIC as model selection criterion selects the portfolio which optimizes an unbiased estimate of the out-of-sample Sharpe ratio.

SRIC punishes less than AIC for additional parameters. The intuition is that a (naive) mean-variance investor will take on too much risk because of over-optimistic in-sample estimates. AIC correctly punishes for that. In contrast, a Sharpe ratio investor is not affected by the absolute risk exposure but only the risk return trade-off. 

A more realistic mean-variance investor, however, would shrink the in-sample estimates (e.g. using the estimator derived in this paper), take less leverage and hence be as well more interested in the Sharpe ratio rather mean-variance utility. Also for such an investor SRIC would be the right model selection criterion. 

Correcting the in-sample Sharpe ratio for noise fit and estimation error turns out to be simple. Let $\hat \rho $ be the in-sample Sharpe ratio maximized over a $k$-dimensional\footnote{Here, $k$ is the number of parameters that influence the Sharpe ratio, i.e. the leverage of a portfolio is not counted. For example the problem of choosing the (ex-post) optimal portfolio out of $k+1$ assets possesses $k$ parameters as one parameter is redundant and only determines the volatility. To be precise, the scale parameter not only determines the volatility but also the sign of the portfolio (long or short). This, however, is a discrete choice and does not matter asymptotically (for large $T$). } parameter space and $T$ years of in-sample data. The main result of this paper is that if we define 
\begin{eqnarray}
\text {SRIC}  = \underbrace{\hat \rho  \vphantom{ \left(\frac{a}{\hat b}\right) }  }_{\text{in-sample fit}}   - \underbrace{ \frac{k}{T \hat \rho}}_{\text{estimated noise fit and estimation error}}
\label{eq_intro_sic}
\end{eqnarray} 
where SRIC stands for \textit{Sharpe Ratio Information Criterion}, then
\begin{eqnarray*}
\E\brae{\text{out-of-sample Sharpe ratio}} = \E\brae{\text {SRIC}}.
\end{eqnarray*} 
In particular, SRIC is an unbiased estimator of the out-of-sample Sharpe ratio. 

Notice the simplicity of the expression, that only involves the in-sample Sharpe, the number of parameters and the length of the in-sample period in years. For example, if there is an investment strategy with $k=5$ parameters and an optimal in-sample Sharpe of $\hat \rho=1$ over $T=10$ years of data, then the estimated out-of-sample Sharpe would be $1-\frac{5}{10}= 0.5$.

\begin{figure}[H]
	\centering
			\includegraphics[width=8cm]{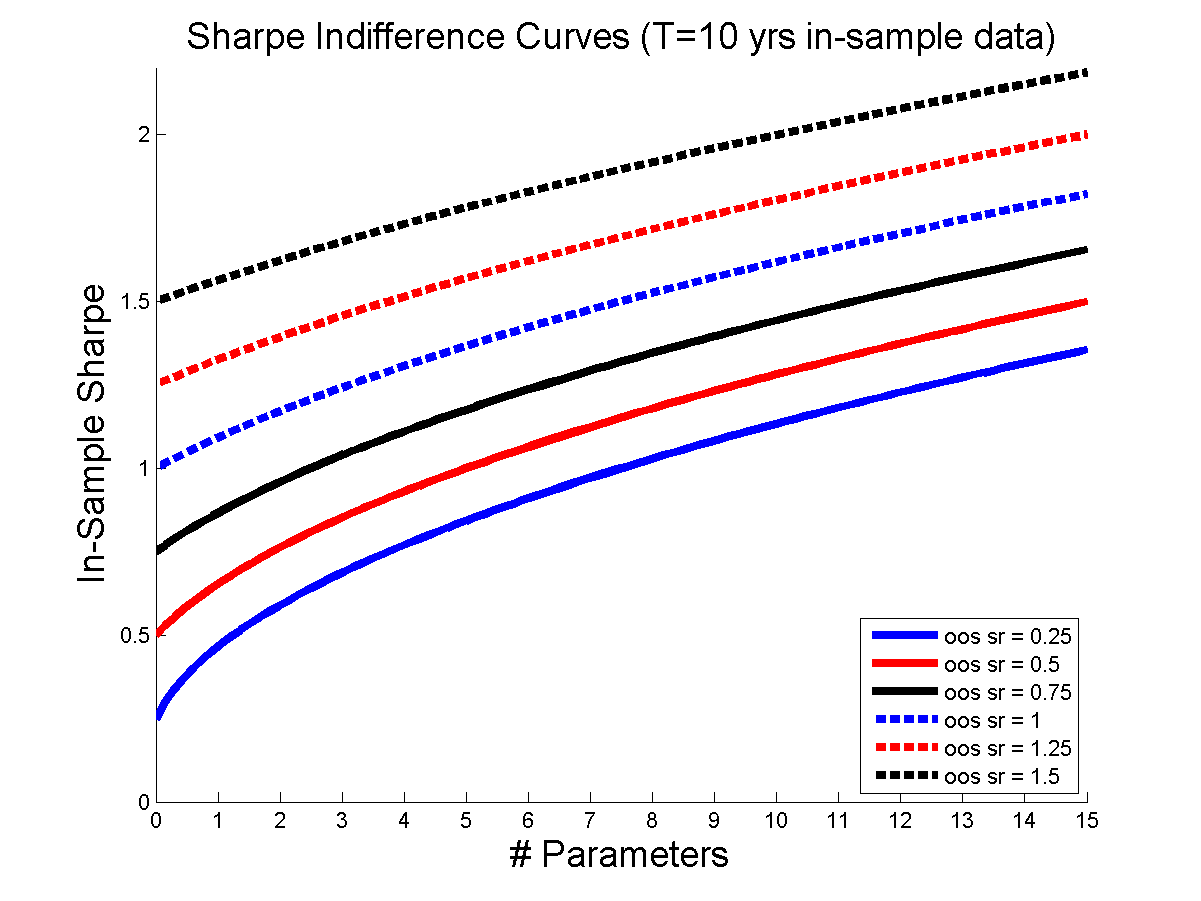}
		\flushleft
	{\textbf{Figure 1 (Sharpe Indifference Curves):} In-sample Sharpe ratios that are needed for a fixed estimated out-of-sample Sharpe depending on the number of parameters assuming $10$ years of in-sample data.}
	\label{fig1_sharpe_indiff}
\end{figure}

Figure 1 illustrates SRIC as a model selection criterion as well as an informant about out-of-sample performance in form of \textit{Sharpe Indifference Curves}, i.e. combinations of measured in-sample Sharpe ratios and number of parameters that lead to the same estimated out-of-sample Sharpe ratio (for $T=10$ years of in-sample data). 

A special case of model selection lies in portfolio optimization. It has been documented that on a variety of data sets the (naive) portfolio optimization using in-sample estimates for mean and covariance to build the Markowitz portfolio underperforms simple benchmarks like the equally weighted portfolio (see, e.g. \cite{demiguel2009optimal}), the cause of which has been attributed to estimation error. 
On a high-level, the question of whether to choose the equally weighted portfolio or the Markowitz portfolio or something in-between, can be seen as the question on how many basis vectors (e.g. principal components) to optimize upon. If one uses only one basis vector, like equal weights, there is little estimation risk, but if one optimizes over a complete basis one obtains the Markowitz portfolio and larger estimation risk, see, e.g. \cite{chen2016efficient}.


SRIC puts a price on noise fit and estimation error in terms of the Sharpe ratio and thereby makes the trade-off between in-sample fit and estimation error precise. Namely, the reason for why in many empirical data sets the equally weighted portfolio outperforms the Markowitz portfolio is two-fold. One is estimation error. The other is how much one gets in return for it. Often, most of the return opportunity is already captured by the equally weighted portfolio (or the first principal component for that matter). Additional portfolio directions offer little added return opportunities, often too small to justify the increase in estimation risk. When is the trade-off beneficial to the Sharpe ratio and when not? This question is answered by SRIC.
%

As an illustration we look at the set of $10$ industry portfolios from Kenneth French's website also used by \cite{demiguel2009optimal}. While we reconfirm that for estimation horizons of $5$ to $10$ years the Markowitz portfolio is inferior to the equally weighted portfolio, we find that this no longer holds true for estimation horizons of $2$ years and less. Though this is mainly due to first half of the sample, it illustrates our point that estimation error, being even larger for smaller estimation horizons, cannot be the full story.

The intuition is that the Markowitz portfolio relies on the sample mean and therefore effectively is a trend or momentum portfolio. While there are no trends on higher components on a time scale of $5$ to $10$ years, there are trends on a time scale of $2$ years or less which are worth paying the estimation error. Which model to choose, a simple model like the equally weighted portfolio or a more complex one like the Markowitz portfolio, can be decided by SRIC.


Additional things to highlight when it comes to summing up the contributions are:  i) unlike AIC (at least directly) SRIC allows to compare models estimated over different in-sample periods. ii) We also derive uncertainty bounds for noise fit and estimation error and iii) due to the close relation to AIC, this paper provides an interpretation of the latter in terms of noise fit and estimation error. At least the authors have learned a lot about AIC while thinking about this problem.

For the ease of exposition we restrict this article to the case where the parameters influence returns linearly. However, the results are still valid asymptotically in a more general setting of non-linear dependence. 

We shall also mention the $R$ package by Steven Pav (\cite{pav2015package}) and the papers surrounding it (\cite{pav2014portfolio, pav2015inference, pav2015notes}) which provide tools and theorems for statistical inference (noise fit) on the true Sharpe ratio  exploiting the fact that the squared Sharpe ratio follows a $\chi^2$-distribution (respectively $F$-distribution, when the covariance is estimated).

The next Section~\ref{section_setup} introduces the formal setup. Section~\ref{section_linear_case} formulates the main theorem. Section~\ref{section_model_selection} discusses the results in the context of model selection and compares to the Akaike information criterion. Section~\ref{section_other_application} illustrates the results in various examples and toy applications. Section~\ref{section_extensions} discusses extensions before we conclude in Section~\ref{section_conclusion}. Most proofs are relegated to the appendix.



\section{Set-Up}
\label{section_setup}



Fix a probability space $(\Omega, (\mathcal F_t)_{t\geq 0}, \P)$. Let $\Theta= \R^{k+1}$ be a (k+1)-dimensional parameter space\footnote{We use the notation $k+1$ as one dimension will only influence the volatility and only $k$ parameters will influence the Sharpe ratio. }. Let $r_t\in \R^{k+1}$, $t \in [0,T]$ be a (k+1)-dimensional series of returns over $T$ years. Be it asset returns, factor returns or returns associated with a specific predictor variable. We are interested in the Sharpe ratio of the returns $s^{\theta}_t = r_t \theta$. That is where returns are linearly parametrized by $\theta$. Call $[0,T]$ the \textit{in-sample period}. 

Denote by $\hat \mu \in \R^{k+1}$ the (annualized) \textit{estimated} mean return of $r$ and $\Sigma$ its (annualized) covariance matrix. Then $\hat \mu^T \theta$ is the estimated mean return of $s^{\theta}_t$ and $\theta^T \Sigma \theta$ its variance. We assume that $\hat \mu$ is a noisy observation of the true mean return $\mu$, that is $\hat\mu = \mu + \nu$ where $\nu$ is a random variable with covariance matrix $\frac{1}{T} \Sigma$. Here we assume that $\nu$ is normally distributed\footnote{This assumption eases the analysis and is not too restrictive. Here, the quantities of concern are $T$-year return averages. So even with non-normal returns, thanks to the central limit theorem, its $T$-year averages are close to being normal. All results, however, are still asymptotically valid (for large $T$) for other than the normal distribution when moments are sufficiently bounded.  }. We further assume that $\Sigma$ has full rank.

Note that we made the assumption that $\nu$ has the same covariance matrix as $r$ scaled by $\frac{1}{T}$. This is the case if $\hat \mu$  is the \textit{realized} mean return over the in-sample period.

While the true returns are observed with noise, we assume that the covariance matrix $\Sigma$ can be observed without error. This is for instance true in a continuous time framework. But even without continuous return observations it is \textit{close to the truth} in reality as the estimated covariance is orders of magnitudes more precise than the estimated mean return provided the number of parameters is not \textit{too large}. For instance in Example \ref{example_regression} below, even with $500$ stocks, we only need the covariance matrix of a few predictor variables over years of observed data. We do not need the covariance of the $500$ stocks themselves. We will comment on this later. 

Applications of the set-up lie wherever portfolio returns are linearly parametrized  and are therefore ample within portfolio and asset management.  The purpose is often to predict returns and maximize the out-of-sample Sharpe ratio.  In the most simple case the vector $\theta$ could describe portfolio weights on assets or (tradable) risk factors $i=1,\ldots,k+1$.  We illustrate this in Example \ref{example_portfolio_choice}.


\begin{example}[Portfolio Choice]\label{example_portfolio_choice} Let there be $k+1$ return streams $r^i_t$.  Let  $$s_t^{\theta} = \sum\limits_{i=0}^{k} \theta_i r_t^i$$ be the return of a portfolio with weights $\theta$. With $\hat \mu$ the annualized realized mean return of $r^i$, $\Sigma$ its covariance, $\mu$ the (unknown) true mean and $T$ the observation length (all annualized), the setting fits the set-up described above. 
\end{example}

To be clear, we do not envision $r_t^i$ to be $500$ different stocks of the S\&P500. We rather think of $r_t^i$ in Example \ref{example_portfolio_choice}  as being the returns of factors or base portfolios like the first few principal components, long/short portfolios on stocks sorted by characteristics like price-dividend ratio, book-value, or momentum and where the optimal factor portfolio can be estimated over years of data.  

The main application we have in mind is estimating the out-of-sample predictive power of linear predictors. 

\begin{example}[Regression, Practitioner's perspective]\label{example_regression} Let there be $i=1,\ldots,N$ markets with return $y^i_t$. Consider the following linear model
\begin{equation}
y^i_t = \sum\limits_{j=0}^k x^{i,j}_t \theta_j + \eps_t^i \label{eq_ez22}
\end{equation} where the mean return is linear in $k+1$ exogenous predetermined characteristics $x^{i,j}$  and the random errors are $\eps_t^i$. Let  $S_t\in \R^{N,N}$ be the known or unknown covariance of $y_t$. For example, we could try to predict the return of the $N=500$ stocks in the S\&P500 Index via characteristics like price earning ratio, price dividend ratio, momentum or other variables. Now consider a practitioner that bets with weights $w_t = \hat S_t^{-1} x_t \theta \in \R^N$ on the return $y_t=(y^1_t,\ldots,y^N_t)$ where $\hat S_t$ is a weighting matrix\footnote{Ideally it is equal to the (unknown) market covariance, hence the notation. However, any (!) weighting scheme, e.g. equal weights, will do as long as it is deterministic (or predetermined at time $t$ for practical purposes). The set-up at hand treats~$\hat S_t$ as exogenous. Our analysis will be \textit{conditional} on the weighting $\hat S_t$. That is we are interested in the Sharpe ratio that an investor obtains as a combined choice of a weighting scheme $\hat S_t$ and parameter $\theta$, where she maximizes over the latter.} that maps the predictions $x_t\theta$ into portfolio weights. By that she receives the following series of returns 
\begin{eqnarray*}
y_t^T w_t &=& y_t^T \hat S_t^{-1} x_t \theta\\
&=&r_t^T \theta \quad \text{ with $r_t^T= y_t^T \hat S_t^{-1} x_t $}.
\end{eqnarray*}
With returns $r_t$, $\hat \mu$ its annualized (realized) mean return and $\Sigma$ its covariance, the setting is as described in this section. It is important to note that $\Sigma\in \R^{k+1, k+1}$ is of dimension $k+1$, the number of predictors, which is usually much lower than $N$, the number of stocks and that the regression will usually be performed over long time horizons. The practitioner can optimize the in-sample Sharpe ratio over $\theta$ and, applying the estimator derived in this paper, get an unbiased estimate of the out-of-sample Sharpe ratio. 
\end{example}

In this paper, like in Example \ref{example_regression}, we take the perspective of a practitioner looking at portfolio returns parametrized by $\theta$. Note that when her weighting matrix $\hat S_t$ differs from the market covariance~$S_t$, the optimal parameter $\hat \theta$ does not necessarily equal the true parameter $\theta$ in equation (\ref{eq_ez22}). But the practitioner does not care about the true parameter, she cares about the parameter that maximizes the Sharpe ratio of her portfolio. 

However, when the market covariance is known\footnote{Actually an estimate would work as well, if the estimate is such that the implied covariance in $\theta$ equals the realized covariance. More precisely, we need that the quadratic term in the GLS equals the true $(k+1)$-dimensional covariance matrix in $\theta$, namely:
$$\sum\limits_t x_t^T \hat S_t^{-1} x_t =  \sum\limits_t x_t^T \hat S_t^{-1} S_t \hat S_t^{-1}x_t$$
This is a consistency condition on the market covariance estimates $\hat S_t$.}, we can give these portfolios a generalized least square (GLS) interpretation. Consequently, $\theta$ will be estimated consistently. 

\begin{example}[Regression, Academic's perspective]\label{example_leastsquares} In context of Example \ref{example_regression} a natural way to estimate $\theta$ would be to minimize the sum of squared residuals weighted by the inverse covariance matrix (generalized least square regression):
\begin{eqnarray*}
\hat \theta = && \operatorname{arg} \min\limits_{\theta} \sum\limits_t \bra{y_t - x_t\theta }^T  S_t^{-1} \bra{y_t - x_t\theta }
\end{eqnarray*}
where $S_t$ is the market covariance. Now
\begin{eqnarray*}
&&\min\limits_{\theta} \sum\limits_t \bra{y_t - x_t\theta }^T  S_t^{-1} \bra{y_t - x_t\theta }\\
&\Leftrightarrow&   \min\limits_{\theta}  -2\sum\limits_t y_t^T \underbrace{ S_t^{-1} x_t \theta}_{w_t} + 
\sum\limits_t \underbrace{\theta^T x_t^T  S_t^{-1} x_t \theta}_{w_t^T  S_t w_t} \\
&\Leftrightarrow&   \max\limits_{\theta}  2 \hat \mu^T \theta - \theta^T  \Sigma \theta \quad \text{ with $\hat \mu^T = c \sum\limits_t y_t^T  S_t^{-1} x_t$ and $ \Sigma = c \sum\limits_t x_t^T S_t^{-1} x_t$ }
\end{eqnarray*}
where $c>0$ is an annualization factor, $\hat\mu^T \theta$ is the annualized mean return of betting with the weighted predictions $w_t=S_t^{-1} x_t \theta$ on the $N$ markets, that is of $r_t^T= y_t^T S_t^{-1} x_t$, the return associated to $\theta$, and $ \Sigma$ its covariance. Hence the generalized least square regression is the same as maximizing the mean variance utility (and therefore the Sharpe ratio) of the weighted portfolio and vice versa, if the weighting $S_t$ equals the market variance.
\end{example}

Note that $S_t$ is generally unknown and the generalized least square estimator (GLS) requires an estimate of it in the same way as the practitioner in Example \ref{example_regression} needs a weighting scheme $\hat S_t$. 

The practitioner does not need to correctly specify the $N$-dimensional covariances $S_t$, whatever she thinks would be an appropriate weighting scheme that generates returns - be it a heuristic - would be fine. The analysis in this paper, i.e. the estimates of the out-of-sample Sharpe ratio are conditional on her choice of $\hat S$ and therefore not subject to any estimation error in $\hat S$. A bad choice would generate bad returns, as a bad choice within GLS would yield less efficient estimates. 

\color{black}

Example \ref{example_leastsquares} shows that GLS is the same as mean-variance optimization where the translation between the regression's predictions and mean variance optimal weights is given by $w_t = S_t^{-1} \hat r_t$. In particular, the parameter that minimizes least squares also maximizes the Sharpe ratio.

Hence, in terms of parameter estimation least square regression and maximization of the Sharpe ratio are equivalent. They are different, however, when it comes to estimating the out-of sample statistics and model selection. For regression (mean-variance utility) the Akaike information criterion selects the model that maximizes estimated out-of sample fit. For the Sharpe ratio, in contrast, SRIC as defined in Theorem \ref{theorem_linear_case_nf_me} selects the model with the highest estimated out-of-sample fit.

\color{black}

\begin{remark}In a different interpretation, rather than observing $T$ years of data and with $\hat\mu$ as the sample mean, $\hat \mu$ could be a \textit{view} or prior belief of the investor obtained from a different model or a quantified guess like in the Black--Litterman model (\cite{black1991asset}). In this case $\nu$ models the uncertainty around the view and $T$ describes its precision (inverse of the variance).
\end{remark}


We distinguish between the in-sample Sharpe ratio $\rho$ and the out-of-sample Sharpe ratio $\tau$:

\begin{definition}[Sharpe ratio]  \label{def_insample_sr}
The \textit{in-sample} Sharpe ratio of parameters $\theta$ is 
\begin{equation*}
\rho(\theta) = \frac{\hat \mu^T \theta}{\sqrt{\theta^T \Sigma \theta}}. 
\end{equation*}
The (unobserved) \textit{out-of-sample} Sharpe ratio, denoted by $\tau$ (for \textit{true} Sharpe ratio), follows by removing the noise term from the mean returns:
\begin{equation*}
\tau(\theta) = \frac{\mu^T \theta}{\sqrt{\theta^T \Sigma \theta}} 
\end{equation*}
\end{definition}

Now consider an investor who maximizes the Sharpe ratio among all parameters $\theta$. As there is a difference between the in-sample and out-of-sample Sharpe, both maximizers are different. We use the following notation:
\begin{notation}\label{not_theta_thetastar}
Denote by $\hat \theta $ a vector of parameters that maximizes the in-sample Sharpe ratio and let $\theta^* $ be a parameter maximizing the out-of-sample Sharpe ratio\footnote{Note that $\hat \theta$ and $\theta^*$ are only unique up to multiplication with a positive constant. }. 
\begin{equation*}
\hat \theta \in \arg \max\limits_{\theta \in \Theta} \rho(\theta), \quad  \theta^* \in \arg \max\limits_{\theta \in \Theta} \tau(\theta)
\quad 
\end{equation*}
%
We abbreviate $\hat \rho= \rho(\hat \theta)$, $\hat \tau=\tau(\hat \theta)$, $ \tau^*=\tau(\theta^*)$,  $\rho^*=\rho(\theta^*)$ 
\end{notation}

We summarize the notation in Table 1.

\begin{figure}[H]
\centering
\begin{table}[H]
\centering
\begin{tabular}{c|c|c|c}
\hline
Value &Symbol & Parameter  & Description\\ \hline \hline
$\rho(\hat \theta )$ & $\hat \rho$ & $\hat \theta$  & Sharpe ratio of optimal in-sample \\
& & & parameter applied to in-sample data set \\ \hline
$\rho(\theta^* )$ & $ \rho^*$ & $\theta^*$  & Sharpe ratio of optimal out-of-sample \\
&&& parameter applied to in-sample data set \\ \hline

$\tau(\hat \theta )$ & $\hat \tau$ & $\hat \theta$ & Sharpe ratio of optimal in-sample \\
&&& parameter applied to out-of-sample data set \\ \hline
$\tau( \theta^* )$  & $\tau^*$ & $\theta^*$  & Sharpe ratio of optimal out-of-sample \\
&&& parameter applied to out-of-sample data set \\ \hline
\end{tabular}
\end{table}
\flushleft
\textbf{Table 1:}  All four combinations of in- and out-of-sample Sharpe and true and estimated parameter 
\label{table_1}
\end{figure}

With Notation~\ref{not_theta_thetastar} we get the decomposition that is central to this paper:
\begin{eqnarray}
\underbrace{\hat \tau \vphantom{\bra{\hat \rho}}}_{ \text{oos Sharpe}}= \underbrace{\hat \rho \vphantom{\bra{\hat \rho}}}_{\text{is Sharpe}} - \underbrace{\bra{\hat \rho - \rho^*} }_{\text{noise fit}} - \underbrace{\bra{ \tau^*-\hat\tau}}_{\text{estimation error}} + \underbrace{\tau^* - \rho^* \vphantom{\bra{\hat \rho}} }_{\text{noise}}
\label{eq_decomposition}
\end{eqnarray}

Decomposition (\ref{eq_decomposition}) says that the out-of-sample Sharpe ratio equals the in-sample Sharpe ratio minus three terms. First, the difference in Sharpe ratio between the estimated and the true parameter on the in-sample set, which can be interpreted as noise fit. Second, the difference in Sharpe ratio between the estimated and the true parameter on the out of-sample set, that is estimation error. Third, the difference in Sharpe ratio of the true parameter between the in-sample and the out-of-sample data set, which is the noise in the in-sample data.

We record this decomposition in the next definition.

\begin{definition} \label{def_noise_fit_sr}
The in-sample Sharpe ratio can be decomposed into
\begin{equation*}
\hat \rho = \hat \tau + \N+\EE+\U
\label{eq_decomposition2}
\end{equation*} with $\E\brae{\U}=0$,
where the following definitions are applied
\begin{eqnarray*}
\mathcal N &=& \rho(\hat \theta) - \rho(\theta^*) \nonumber \quad \text{(Noise Fit)}, \\
	\EE &=& \tau(\theta^*) - \tau(\hat \theta)\nonumber  \quad \text{(Estimation Error)}, \\
\U &=&\rho(\theta^*) - \tau(\theta^*)  \nonumber \quad \text{(Noise)}.
\end{eqnarray*}
\end{definition}

Naturally two questions emerge. First, (1) how informative is $\hat \rho$, the maximal in-sample Sharpe ratio,  for $\tau^*$ the true optimum? Choosing the (ex post) optimal parameter $\hat \theta$ will lead to some noise fit and therefore its Sharpe ratio will (in expectation) overestimate the true Sharpe ratio. But by how much? And second, if the estimated parameter $\hat \theta$, rather than the true parameter $\theta^*$, is applied out-of-sample, what Sharpe ratio can be expected. That is (2) how high is the degradation in Sharpe ratio due to the combination of noise fit and estimation error?

The first question needs a Taylorization and is technically more involved which is why we only remark about it here. 
We answer the second question in Section~\ref{section_linear_case} below. Before that, we comment on our assumption of a known covariance matrix.

When optimizing the Sharpe ratio, there are two different sources of noise. Noise in the estimated covariance matrix and noise in the estimated returns. This corresponds to two different kinds of asymptotics. If we fix the time horizon and increase the sampling frequency, e.g. from monthly over daily to hourly, we get more and more accurate estimates of the covariance. The estimates of the mean returns, however, do not become more precise. In the limit of continuous-time observations the covariances are observed without error and thus known, while the mean returns are not. If, on the other hand, we increase the time horizon,  return estimates become less noisy as well.  

Here, we focus on the case in which the covariance matrix is known and only the mean returns are estimated with noise. This is the relevant case when estimation error in the mean return is present, e.g. in the regression setting in which a few predictive parameters are optimized based on several months or years of in-sample data. 
The reason is that noise in estimated mean returns is so severe that in any realistic application one only optimizes over relatively few assets or parameters and relatively long estimation horizons. 


If, however,  mean returns are imposed (e.g. set equal for all assets as in the \textit{minimum variance} portfolio or set to analyst forecasts with tight confidence bounds) then estimation error in the covariance matrix dominates and the set-up at hand is not directly applicable. Let us remark, however, that if we have a noisy estimate $\hat \Sigma$ of the true covariance $\Sigma$ and we know the true return~$\mu$, we can write for the Markowitz portfolio 

\begin{equation}
w_{\operatorname{MV}}= \hat \Sigma^{-1} \mu =  \Sigma^{-1} \underbrace{\bra{\Sigma \hat \Sigma^{-1} \mu}}_{=\hat \mu} =\Sigma^{-1} \hat \mu.
\label{eq_noise_cov_mean_equivalent}
\end{equation}

Hence the weights are \textit{as if} we knew the true covariance but would apply it to a noisy estimate $\hat \mu$ of returns. From this perspective, noise in the covariance matrix and noise in the mean returns are very similar.

\section{Main Theorem}
\label{section_linear_case}

We are now ready to formulate our main theorem.

\begin{theorem}[Estimation Error and Noise Fit combined for the Sharpe ratio]\label{theorem_linear_case_nf_me} Let $k\geq 1$.  Suppose $\hat \mu=\mu +\nu$ where $\nu \sim \N(0, \frac{1}{T}\Sigma)$ is normally distributed, and $\Sigma$ has full rank, then it holds
\begin{equation*}
\E\brae{ \N+\EE+\U} = \E\brae{ \frac{k}{ T \hat \rho} }.
\end{equation*} In  particular, we have
\begin{equation}
\E\brae{\hat \tau} = \E\brae{ \hat \rho  -\frac{k}{T \hat \rho} }.
\label{eq_lin_noisefit_estimation2} 
\end{equation}

\end{theorem}
\begin{proof}
See Appendix.
\end{proof}
Theorem \ref{theorem_linear_case_nf_me} shows that $SRIC = \hat \rho - \frac{k}{T \hat \rho} $
is an unbiased estimator of the expected out-of-sample Sharpe ratio $\E\brae{\hat \tau}$. We give it the name SRIC for \textit{Sharpe Ratio Information Criterion} for reasons that become clear in Section \ref{section_model_selection}.

The bias correction can be split into noise fit and estimation error. 
\begin{theorem}\label{theorem_nf_ee_split} Grant the assumptions of Theorem \ref{theorem_linear_case_nf_me}.
Assume that $\tau^*>0$, then the bias correction in (\ref{eq_lin_noisefit_estimation2}) can be split into noise fit and estimation error as follows
\begin{eqnarray*}
\hat \tau \approx \text {SRIC} = \underbrace{\hat \rho  \vphantom{ \left(\frac{a^{0.3}}{b}\right) }  }_{\text{in-sample fit}} - \underbrace{ \frac{k}{2 T \hat \rho}}_{\text{estimated noise-fit}}- \underbrace{\frac{k}{2 T \hat \rho}}_{\text{estimated estimation error}} 
\end{eqnarray*} 
where the split is valid asymptotically of order $o(T^{-1})$.  
 
\end{theorem}
The proof is in the appendix.

We now quantify the uncertainties around $\hat \tau$.

\begin{theorem}\label{theorem_uncertainties}
Let $k\geq0$. Suppose $\hat \mu=\mu +\nu$ where $\nu \sim \N(0, \frac{1}{T}\Sigma)$ is normally distributed, and $\Sigma$ has full rank. For the distribution of the difference between the in-sample Sharpe ratio and the out-of-sample Sharpe ratio holds the following: 
\begin{enumerate}
	\item If the true Sharpe ratio $\tau^*=0$ is zero, so is the out-of-sample Sharpe $\hat \tau$ and we have 
	\begin{equation}
	\hat \rho - \hat \tau  = \hat \rho -0 =\norm{\nu}_{\Sigma^{-1}}\stackrel{\mathcal{D}}{=}  \frac{1}{\sqrt{T}} \chi(k+1)
	\label{eq_confidence_set1}
	\end{equation}
	where $\chi(k+1)$ denotes the $\chi$-distribution (square root of $\chi^2$-distribution) with $k+1$ degrees of freedom and $\norm{x}_{\Sigma^{-1} }=\sqrt{x^T\Sigma^{-1}{x}}$.
	\item If $\tau^*>0$, we have, first, 
		\begin{equation}
	\hat \rho - \hat \tau  \leq \norm{\nu}_{\Sigma^{-1}} \stackrel{\mathcal{D}}{=} \frac{1}{\sqrt{T}} \chi(k+1)
	\label{eq_confidence_set2}
	\end{equation}
	so that the uncertainties are bounded by the the $\chi$-distribution. And, second,
\begin{equation}
\hat \rho - \hat \tau  \stackrel{\mathcal{D}}{=}  \frac{1}{T } \frac{1}{\tau^*} Z + \frac{1}{\sqrt{T}} N +R 
\label{eq_confidence_set}
\end{equation}
where $Z$ is $\chi^2(k)$-distributed, $N$ is an independent standard normal distribution and $R$ a remainder such that $\E\brae{T \abs{R}^p}\to 0$ as $T\to\infty$ for all $p\geq 1$.
\end{enumerate}
\end{theorem}

Equation (\ref{eq_confidence_set1}) is convenient under the null hypothesis of a true Sharpe ratio of $0$ which can be tested by using that $T\hat\rho^2$ follows a $\chi^2 $ distribution. Equation (\ref{eq_confidence_set}) has an intuitive interpretation. While the $N$-term corresponds to the noise in the in-sample data set, the $Z$-term corresponds to noise fit and estimation error. With a higher true Sharpe ratio $\tau^*$, any deviation from the true parameter becomes more costly so that the parameters are estimated more precisely and the $Z$-term becomes less relevant.

We conclude this section by proving an analogous result for mean variance utility rather than the Sharpe ratio as measure of fit. Though results are not new, for instance \cite{kan2007optimal} and \cite{hansen2009sample} show results closely related to ours, we find it instructive to present them in a way consistent with the set-up here. This is in particular useful when we later relate SRIC to the Akaike Information Criterion (AIC). 

For this let
\begin{equation*}
\hat u(\theta) = 2 \hat\mu^T \theta - \gamma \theta^T \Sigma \theta \nonumber
\label{eq_mv_utility_is}
\end{equation*}
be the in-sample and
\begin{equation*}
 u(\theta) = 2 \mu^T \theta - \gamma \theta^T \Sigma \theta \nonumber
\label{eq_mv_utility_os}
\end{equation*}
the out-of-sample mean variance-utility. The parametrization is such that for $\gamma=1$ holds $\hat u (\hat \theta) = \rho(\hat\theta)^2$, the squared in-sample Sharpe ratio.

The analogue of Definition \ref{def_noise_fit_sr} is
\begin{definition} \label{def_noise_fit_mv} The in-sample mean-variance utility can be decomposed into
\begin{equation*}
	\hat u (\hat \theta)  = u(\hat\theta)+\mathcal N_{MV} +\EE_{MV}+ \U_{MV}\nonumber
	\label{eq_noisefitmv1}
\end{equation*}
with $\mathcal N_{MV} = \hat u(\hat \theta) - \hat u(\theta^*)$ (Noise Fit),  $\EE_{MV} =u (\theta^*) - u (\hat \theta)$ (Estimation Error), $\U_{MV} =\hat u (\theta^*) - u(\theta^*) $ (Noise).
\end{definition}

Now it is easy to show
\begin{theorem}[Noise-Fit and Estimation Error for Mean-Variance]
\label{theorem_linear_case_squared_nf_me} 
Grant the assumptions of Theorem \ref{theorem_linear_case_nf_me}. It holds:
\begin{equation*}
\E\brae{ \N_{MV}} = \frac{k+1}{ \gamma T }
\label{eq_square_noisefit_estimation} \nonumber
\end{equation*}  and
\begin{equation*}
\E\brae{ \N_{MV} + \EE_{MV}+\U_{MV}} = \frac{2(k+1)}{ \gamma T }.
\label{eq_square_noisefit_estimation2} \nonumber
\end{equation*} 
In particular, $\hat u (\hat \theta)-\frac{2 (k+1)}{\gamma T} = \frac{1}{\gamma} \hat \rho^2-\frac{2 (k+1)}{\gamma T}$ is an unbiased estimator for the out-of-sample utility $u(\hat \theta)$.
\end{theorem}

We will later show that the Akaike Information Criterion (AIC) can (after a linear transformation) be interpreted as $\hat u (\hat \theta)-\frac{2 (k+1)}{\gamma T}$ and thus as an unbiased estimator of the out-of-sample utility.

%
%
%
%

\section{Model Selection (Sharpe Information Criterion)}
\label{section_model_selection}

\subsection{Sharpe Information Criterion}

\label{section_sharpe_information_criterion}

In the previous sections, the objective was to quantify noise fit and estimation error in order to gain insights about the out-of-sample Sharpe ratio. In this section we apply the results to model selection. 
Let $\Theta= \Theta^1 \dot{\cup} \ldots \dot{\cup} \Theta^n$ with $\Theta^i \subset \R^{k_i+1}$ be a family of parameter spaces. Model selection is about selecting a pair $(i, \theta_i)$, that is a model $i\in \{1,\ldots,n\}$ and a parametric fit $\theta_i  \in \Theta^i$. A typical goal of model selection is to choose the model with the highest out-of-sample fit\footnote{Some criteria have other goals, e.g. the Bayesian Information Criterion (BIC) maximizes the asymptotic posterior probability of choosing the \textit{true} model. }. 

In this paper, we measure \textit{fit} by the out-of-sample Sharpe ratio. Naturally, the out-of-sample Sharpe ratio $\hat \tau(\theta_i)$ is a a random variable whose distribution depends on unknowns like the true parameters. However, by Theorem \ref{theorem_linear_case_nf_me}, SRIC is an unbiased estimator for the out-of-sample Sharpe ratio and can be calculated on observables like the in-sample Sharpe ratio only. This suggests to choose the model $i$ with the highest SRIC
\begin{equation}
SRIC^i = \rho(\hat \theta_i) - \frac{k_i}{T \rho(\hat \theta_i) }
\label{eq_sric_modelselection}
\end{equation}
where $\hat \theta_i$ maximizes $\rho$ over $\Theta^i$. 

SRIC as selection criterion is justified by Theorem \ref{theorem_linear_case_nf_me} as long as estimates of the out-of-sample Sharpe ratio are the variable of interest (that is the first moment of the distribution of out-sample Sharpe ratios).  For higher moments, matters get more complicated. We refer to Theorem \ref{theorem_uncertainties} where we derive the asymptotic distribution of out-of-sample Sharpe ratios. 

We now show that SRIC is exactly analogous to the Akaike Information Criterion (AIC), with the difference that the latter uses log-likelihood rather than the Sharpe ratio as measure of fit. 

\subsection{Relation to AIC}
\label{section_relation_to_aic}


To derive the Akaike Information Criterion (AIC), see \cite{akaike1974new},  we need to associate $\theta \in \Theta$ with a prediction and derive its log-likelihood. That is we need to underpin our set-up with a predictive model first before we can attach a meaning to log-likelihood in this context. This is done in the appendix where we show

\begin{theorem}\label{theorem_loglikelihood}
For a suitable underlying predictive model and for an appropriate reference measure~$\P^0$:
\begin{eqnarray*}
\frac{2}{T}\log \frac{d\P^{\theta}\brae{ (p_t)_{t \in [0,1]} }}{d\P^0\brae{ (p_t)_{t \in [0,1]} } } &=&  2 \hat \mu^T \theta -   \theta^T \Sigma\theta.
\end{eqnarray*}
\end{theorem}
 
In particular, (not surprising in the Gaussian case) the log-likelihood is a multiple of mean-variance utility $\hat u$. Therefore, maximizing the former is equivalent to maximizing the latter.

The AIC is defined as 
\begin{equation*}
AIC = -2\text{log likelihood }+ 2 \cdot \text{number of parameter}
\label{eq_def_AIC}
\end{equation*}
so that we obtain by use of Theorem \ref{theorem_loglikelihood}
\begin{eqnarray*}
AIC &=& -2 \log \frac{d\P^{\theta}}{d\P^0} + 2\bra{k+1} = - T \hat \rho^2 + 2 (k+1)
\end{eqnarray*} and, after a linear transformation, our preferred normalization of the AIC
\begin{eqnarray}
\frac{- AIC}{T} &=& \hat \rho^2 - \frac{2(k+1)}{T}. \label{eq_aic_sharpe2}
\end{eqnarray}
Equation (\ref{eq_aic_sharpe2}) shows that we can express the AIC in terms of mean-variance utility or the squared Sharpe ratio. In combination with Theorem \ref{theorem_linear_case_squared_nf_me} this yields the interpretation of the AIC as expected out-of-sample mean-variance utility or log-likelihood for that matter, that is AIC corrects the in-sample mean-variance utility for noise fit and estimation error.  

%
%
%
%



\begin{remark}  \label{remark_sricdimension}
SRIC chooses a higher dimension than AIC, but both converge towards each other for $T\to \infty$.\footnote{To see the first statement, note that
\begin{eqnarray*}
\frac{\frac{\partial SRIC}{\partial \rho}}{-\frac{\partial SRIC}{\partial k}}  =  T \rho+\frac{k}{\rho}
\geq   T \rho =\frac{\frac{\partial AIC}{\partial \rho}}{-\frac{\partial AIC}{\partial k}}.
\end{eqnarray*}
To see the second statement, observe
\begin{eqnarray*}
 T (SRIC)^2 &=& T \bra{\hat \rho -\frac{k}{T \hat \rho} }^2 
= T \hat \rho^2 - 2k + \bra{\frac{k}{ \hat \rho}}^2 \frac{1}{T} 
= -AIC + const + \bra{\frac{k}{ \hat \rho}}^2 \frac{1}{T}.
\end{eqnarray*}
Hence the difference between $-\operatorname{AIC}/T$ and SRIC is of order $o(T^{-1})$.} 
The intuition is that a mean-variance investor will be punished by loading on too much risk exposure (leverage) on overestimated in-sample estimates which AIC correctly penalizes for. The Sharpe ratio, by contrast, is not influenced by excessive leverage and therefore does not need this additional penalty. 
\end{remark}





%
%

\section{Applications}
\label{section_other_application}

Theorem \ref{theorem_linear_case_nf_me} has ample applications. Those lie wherever estimates of the out-of-sample Sharpe ratio are of interest, be it in portfolio or asset management, when searching for quantitative trading strategies, or within academia when asking whether in-sample anomalies are strong enough to be viable out-of-sample.  

Its first use is as an unbiased estimator for the out-of-sample Sharpe ratio. In Section~\ref{applications_oos1}, we illustrate this. As the theorem claims, the average simulated out-of-sample Sharpe ratio agrees with the average SRIC. 

Its second use is as model selection criterion analogous to the Akaike information criterion. Figure~1, already showed Sharpe Indifference Curves, that is combinations of in-sample Sharpe ratios and number of parameters that yield the same estimated out-of-sample Sharpe ratio. In the next sections we illustrate SRIC as model selection criterion in simulations (Sections \ref{applications_oos2} and \ref{applications_oos3}) and on real data (Section \ref{applications_industry}). For completeness we also show the application of SRIC in a regression setting (Sections~\ref{applications_reg1} and \ref{applications_reg2}). 

Note that assessing the performance of a model selection criterion is not straight forward and always depends on judgment. This is as its performance depends on the unknown true model. Any selection criterion which is biased towards the unknown true model has a head start. This problem does not vanish for simulations or large collections of data sets. Any simulation pins down a prior probability distribution for the true model and the prior determines the optimal bias for the model selection criterion. 

As an attempt to overcome this difficulty we i) choose simulation parameters that are neither biased towards a high- or low-dimensional model and ii) we vary a fixed true model over a range of dimensions and claim that SRIC does well over the broad range of true models even though  at the extremes a criterion which is biased towards the true model outperforms. 

\color{black}

\subsection{Simulation 1: Estimating Out-Of-Sample Sharpe Ratios}
\label{applications_oos1}

We illustrate Theorem \ref{theorem_linear_case_nf_me}. For this, we simulate $T=1,\ldots,25$ years of daily returns ($252$ daily returns per year) with different degrees of freedom and true Sharpe ratios\footnote{E.g. for $T=10$ years and $k=10$ degrees of freedom and a true Sharpe of $\tau^*=1$ we simulate $(r_{t,i})\in \R^{252T,k+1}$ so that $r_{t,i}$ are $\text{iid}$ $N(\mu,\sigma)$-distributed with $\sigma=0.1/\sqrt{252}$ and $\mu=0.1/(252\sqrt{k+1})\tau^*$}. 

We then calculate the optimal in-sample Sharpe ratio, the in-sample Sharpe ratio adjusted for noise fit according to Theorem \ref{theorem_nf_ee_split}, SRIC, and the out-of-sample Sharpe ratio.  

\begin{figure}[H]
\label{fig_k10}
\begin{minipage}[t]{0.5\textwidth}
\includegraphics[width=\textwidth]{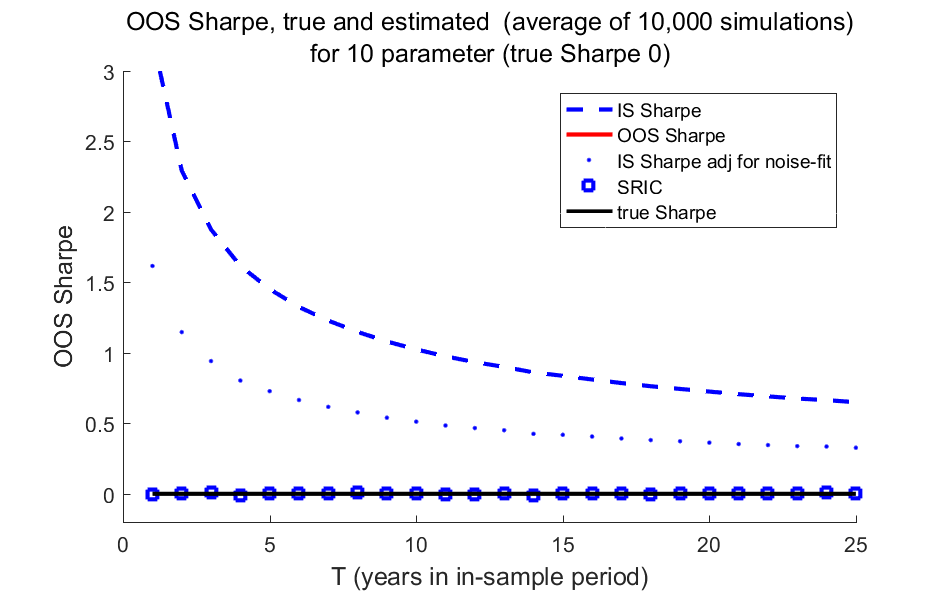}
\end{minipage}
\begin{minipage}[t]{0.5\textwidth}
\includegraphics[width=\textwidth]{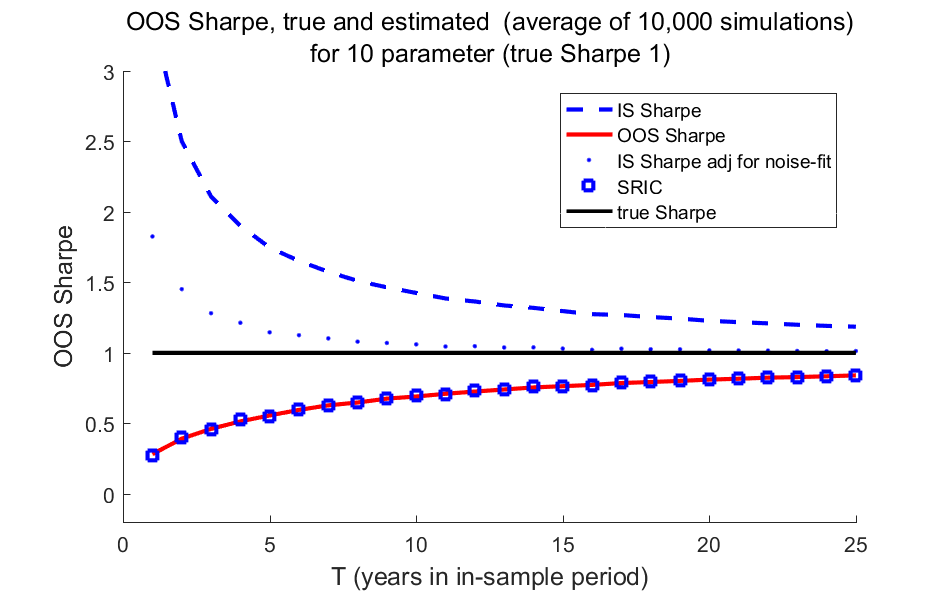}
\end{minipage}
{\textbf{Figure 2:} In-sample Sharpe ratio, adjustments for noise fit and estimation error and out-of-sample Sharpe ratio when the true Sharpe ratio is $\tau^*=0$ (left hand side) or $\tau^*=1$ (right hand side) when the number of parameters is $k=10$. All numbers are averages over $10,000$ random draws.}
\end{figure}

Figure 2 shows the results when the true Sharpe ratio is $\tau^*=0$ (left hand side) or $\tau^*=1$ (right hand side) for $k=10$ degrees of freedom. All numbers are averages over $10,000$ random draws. As the theorem claims, the in-sample Sharpe ratio adjusted for both, noise-fit and estimation error, fits the out-of sample Sharpe ratio.


\subsection{Simulation 2: Basic Model Selection}

\label{applications_oos2}


In this section we test SRIC as model selection criterion and compare it to the Akaike Information Criterion (AIC) within a simulation. 

We simulate\footnote{That is with $\tau^*=1$ we simulate $(r_{t,i})\in \R^{1260,20}$ so that $r_{t,i}$ are independent $N(\mu_i,\sigma)$-distributed with $\sigma=0.1/\sqrt{252}$ and $\mu_i=0.1/(252\sqrt{20})\tau^*$} $T=5$ years of a $20$-dimensional model with a true Sharpe ratio $\tau^*$ of $1$. The $20$-dimensional model is part of a $40$-dimensional model, the full model, with a true Sharpe ratio of $x \geq 1$. We call the first model the base model and the second model the full model.

We now let SRIC and AIC decide between the base and the full model and record the out-of-sample Sharpe ratio. We average over $10,000$ random draws of $T=5$ years of in-sample realizations while we vary $x$, the true Sharpe ratio of the full model. 

\begin{figure}[H]
	\centering
		\includegraphics[width=18cm]{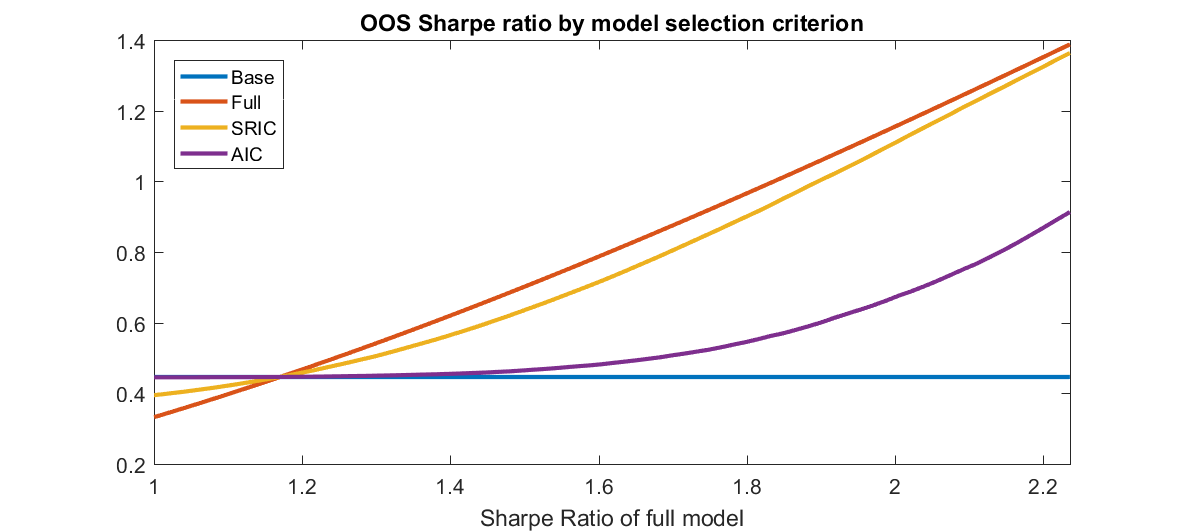}
		\flushleft

	{\textbf{Figure 3:} Average out-of-sample Sharpe ratio for selecting between base and full model depending on true Sharpe ratio of full model. Averages over $10,000$ trials. }
	\label{fig_vary_sr}
\end{figure}

Figure 3 shows the results. With a true Sharpe ratio for the full model of $1.17$ or lower it is optimal to stick with the lower dimensional base model and its true Sharpe ratio of $1$. 

As AIC is biased towards choosing lower dimensional models (see Remark \ref{remark_sricdimension}) it performs better than SRIC in the sub $1.17$ region. However, as the true Sharpe ratio is unknown so is the fact that a lower dimensional model is preferable. With a true Sharpe ratio of $1.17$ or higher, the full model yields a higher out-of-sample Sharpe ratio than the base model and SRIC performs better than AIC. In fact, it takes AIC a substantially higher Sharpe ratio for the full model to see its benefit and overcome its bias towards the smaller base model. 

\subsection{Simulation 3: Extended Model Selection}
\label{applications_oos3}
We now vary the dimension of the true model to see how SRIC performs for a range of different true dimensions.  

For this, we simulate $T=5$ years \`{a} $252$ daily in-sample returns all independent with an annualized volatility of 10\%. We vary the dimension of the true model between $1$ and $100$. For a true model with dimension $k^*$ we give all dimensions $k\leq k^*$ a Sharpe ratio that is uniformly distributed within $[0,0.5]$ and all dimensions $k>k^*$ a Sharpe ratio of $0$\footnote{That is we first draw $\hat x$ uniformly $[0,1]$ and conditional on $\hat x$ we have  $r_{t,i} \sim N(\mu_i,\sigma)$ with $\sigma = 0.1/\sqrt{252}$ and $\mu_i=0.5\hat x/252$ if $i\leq k^*$ and $0$ else, independent across $t=1,\ldots,1260$ and $i=1,\ldots,100$.}. 

Under a candidate model $k$ we understand the portfolio consisting of the first $k$ dimensions $1,\ldots,k$ which is optimal in-sample. We select models $\hat k$ according to SRIC and AIC and denote the corresponding out-of-sample Sharpe ratio. We also note the out-of-sample Sharpe ratios for the full model (Markowitz) and the $1D$ model (which just chooses the first dimension, i.e. $k=1$).

\begin{figure}[H]
	\centering
		\includegraphics[width=18cm]{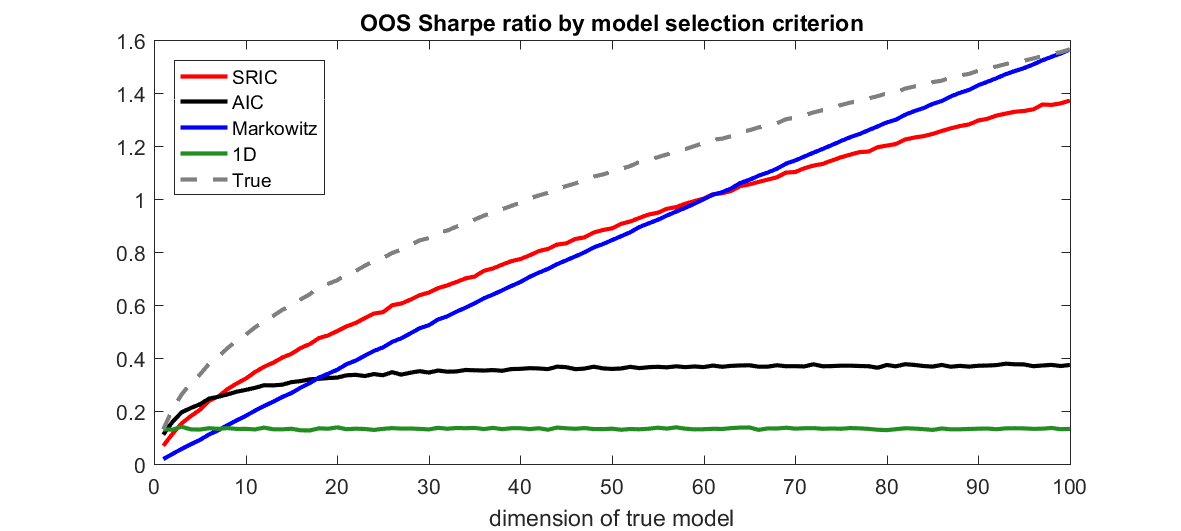}
		\flushleft

	{\textbf{Figure 4:} Average out-of-sample Sharpe ratio depending on dimension of true model for different selection criteria. Averages over $10,000$ trials. }
	\label{fig_vary_dim}
\end{figure}

Figure 4 shows the corresponding average (over $10,000$ draws) out-of-sample Sharpe ratios for different dimensions of the true model by selection criterion.

For comparison, the gray dashed line shows the out-of-sample Sharpe ratio of the true model. Of course the true model is unknown and its Sharpe ratio therefore unattainable. 

As can be seen, SRIC performs well over the full range. It tracks the shape of the performance of the unknown true models. For true dimensions between $7$ and $61$ it has the highest out-of-sample performance among all model selection criteria, that is SRIC, AIC, the $1D$ model and the full model. 

Outside that range, selection criteria which are biased towards the true dimension perform better. When the dimension of the true model is high ($\geq 62$) the full model has the highest out-of-sample Sharpe ratio; but of course as the dimension of the true model is unknown so is the fact that the full model is optimal. When the true dimension is low, the $1D$ model or AIC, both which are biased towards choosing a low dimension, outperform the full model and selection via SRIC. 

Still, SRIC ranks at least second best, except for $k^*=1$ where it ranks behind AIC and the $1D$ model. Most notably, unlike AIC, it benefits from predictability in higher dimensions.

\subsection{Real Data: 10 Industry Portfolios}
\label{applications_industry}

We now illustrate SRIC as model selection criterion for a trend system on the 10 industry portfolios on Kenneth French's website\footnote{\url{http://mba.tuck.dartmouth.edu/pages/faculty/ken.french/ftp/10_Portfolios_Prior_1_0_CSV.zip}}.
We use daily data from 2nd January 1963 to 29th July 2016. First, we subtract the Fed fund rate\footnote{Effective fund rate, daily series, obtained from Federal Reserve \url{http://www.federalreserve.gov}} to get excess returns.  At the beginning of each month we do the following. We look back $lb=1,\ldots,120$ months and compute the principal components on daily data to obtain $10$ factors. To be exact, we choose the equally weighted portfolio as the first factor (rather than the first component of PCA) and then compute the remaining principal components on the space orthogonal to the first factor. We do it this way in order to interpolate between equal weights and the Markowitz portfolio.

By  model $k=1,\ldots,10$ we denote the portfolio that maximizes the in-sample Sharpe ratio on the last $lb$ months over the first $k$ factors $1,\ldots,k$. This way the models $k$ interpolate between the equally weighted portfolio ($k=1$) and the full Markowitz portfolio ($k=10$).  At the beginning of each month, we let SRIC and AIC choose the model $k$ and apply the weights for the subsequent month.  For this, we always scale the portfolio to $10\%$ annualized volatility on the in-sample data (that is the last $lb$ months). We scale to constant volatility because we are concerned with portfolio choice in the \textit{cross section} and do not attempt to \textit{time} the market by having time varying risk exposure. This way we get a series of rolling out-of-sample returns from 1st January 1973 to 29th July 2016.

We vary the lookback $lb$ over which the portfolios are formed. While on longer time scales there might be only one rewarded factor, the market, on shorter time horizons there might be several industry-specific trends. Hence depending on the lookback the dimension of the best-performing model might vary.

\begin{figure}[H]
	\centering
		\includegraphics[width=18cm]{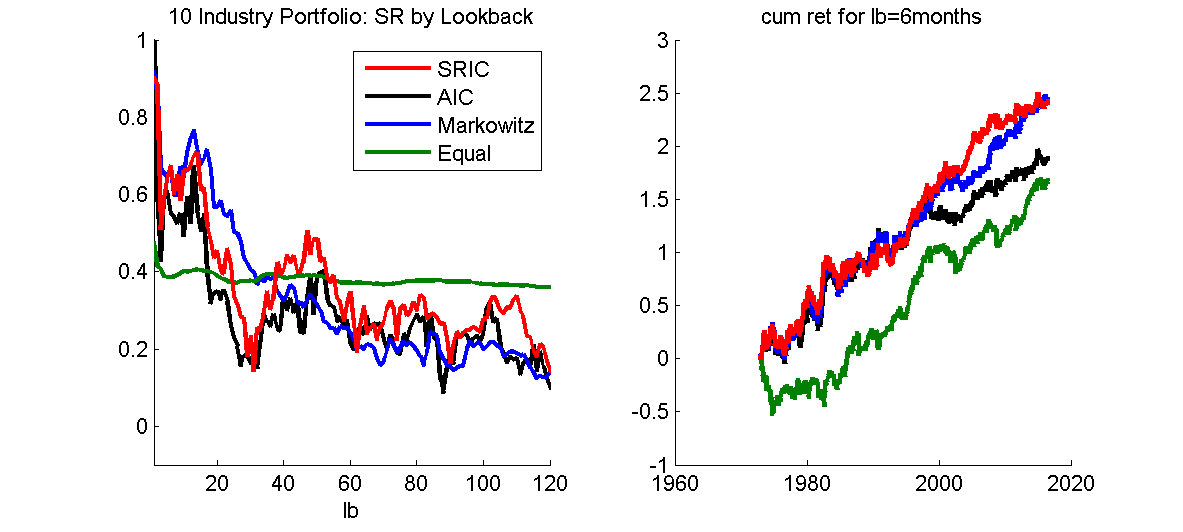}
		\flushleft
	{\textbf{Figure 5 (10 Industry Portfolios):} Left: Sharpe ratio for lookback of 1 to 120 months for different model selection criteria. Right: Cumulative Returns for a lookback of 6 months. The Sharpe ratios are $0.56$, $0.44$, $0.56$ and $0.39$ for selection by SRIC, AIC, Markowitz and Equal Weights.}
	\label{fig_10industry}
\end{figure}

The results are illustrated in Figure 5. The left panel shows the out-of-sample Sharpe ratio for the equally weighted portfolio, the Markowitz portfolio and portfolio choice according to SRIC and AIC depending on the lookback. The right panel shows the equity curves for a lookback of $6$ months. Several things are noteworthy. i) For a long lookback of $lb=36$ months or more, the equally weighted portfolio outperforms the Markowitz portfolio. ii) For short lookbacks of up to $24$ months, however, the Markowitz portfolio outperforms the equally weighted portfolio. 
iii) SRIC dominates AIC as model selection criterion and performs better than the equally weighted portfolio for short lookbacks and better than the Markowitz portfolio for long lookbacks. 

Knowing the results in Figure 5, one would choose the equally weighted model for long lookbacks and the full model for shorter lookbacks. A priori, however, this is unknown and a model selection criterion like AIC or SRIC is needed. 


Note that the Markowitz portfolio (we admit though that the outperformance is due to the first part of the sample) performed better than the equally weighted portfolio for shorter lookbacks, i.e. for lower estimation horizons. This might seem surprising. Namely with a shorter lookback estimation error is higher and hence one could assume that a lower dimensional model like equal weights would work better. Therefore estimation error cannot be the full answer to the question why the equally weighted portfolio outperforms the Markowitz portfolio on many data sets and parameter combinations as for instance for lookback horizons of $3$ years and more. 

Estimation error is only half of the story. The other is that on those data sets most of the return opportunity is already captured by the equally weighted portfolio. In this example, on horizons of $5$ to $10$ years, the return opportunities are concentrated in one component (equal weights). Any additional dimension does not add enough return potential to justify the additional estimation error.  

On shorter horizons, on the other hand, the return opportunities are spread out across more portfolio directions. That is there are time varying industry specific trends on horizons up to $2$ years. Therefore, with shorter estimation windows, one gets something of sufficient value in return for the extra estimation error, while for longer windows one does not.

The question of naive benchmarks versus portfolio optimization is therefore less a question of estimation error rather than of the exact trade-off between estimation error and return opportunities. Theorem \ref{theorem_linear_case_nf_me} prices this trade-off in terms of the Sharpe ratio.

\subsection{Simulation 4: Regression}
\label{applications_reg1}
In the previous subsections we illustrated SRIC in a portfolio choice context where the task was to combine $k$ return streams optimally. Here we illustrate the use of SRIC in a regression context. 

Note that by the reasoning in Remarks \ref{example_regression} and \ref{example_leastsquares} there is no mathematical difference between the two. Hence the following two subsections are more for completeness and as a proof of concept.  In addition, we use them to make the following point: The true model might be suboptimal due to estimation error in its parameters and might be outperformed by a less complex model.     

For this, we create a simulation in which there are decreasing marginal benefits of additional dimensions. In practice, this can be because the predictability of each additional dimension decreases or because the additional dimensions are correlated with the existing ones so that part of their information is already contained in lower dimensions. 

Here is what we do: We draw $30$ normally distributed random variables $x_{i}$, $i=1,\ldots,30$ with pairwise correlation $\rho=0.2$, mean zero and variance $1$. We draw $1260$ independent samples of $x$ and collect them in a $1260$ times $30$ dimensional matrix $X=\bra{X_{t,i}}$. This corresponds to $T=5$ years with $252$ business days each. We then simulate $1260$ daily returns  $y\in \R^{1260,1}$ that can be predicted by $X$. That is we use the model
$$y_t = X_t \beta + \eps_t $$
where $\eps_t$ are iid normal with zero mean and standard deviation $0.1/\sqrt{252}$ which corresponds to an annualized volatility of $10\%$ and $\beta=b (1,\ldots,1)^T$ with $b$ chosen such that each individual predictor $X_i$ ($i=1,\ldots,20$)  has a true Sharpe ratio\footnote{That is $b=0.5/\sqrt{252} \cdot 0.1/\sqrt{252} \cdot f$ with $f=N/(1^T C 1)$, where $C$ is the correlation matrix of $X$  and $N=30$. Hence $b\approx 2.9178 \cdot 10^{-5}$.} of $0.5$. Model $k$ is then to regress on the first $k$ predictors $X_{t,i}$, $i=1,\ldots,k$ as in Example~\ref{example_regression}. That is we look at the return streams $r_{t,i}$ of betting with weights $X_{t,i}$ on $y_t$:
$$r_{t,i} = y_t X_{t,i}.$$
We then derive the optimal in-sample portfolio $\hat \theta$ of returns $r_i$, $i=1,\ldots,k$ and look at their out-of-sample Sharpe ratio depending on the chosen number of predictors $k$. 

\begin{figure}[H]
	\centering
		\includegraphics[width=7cm]{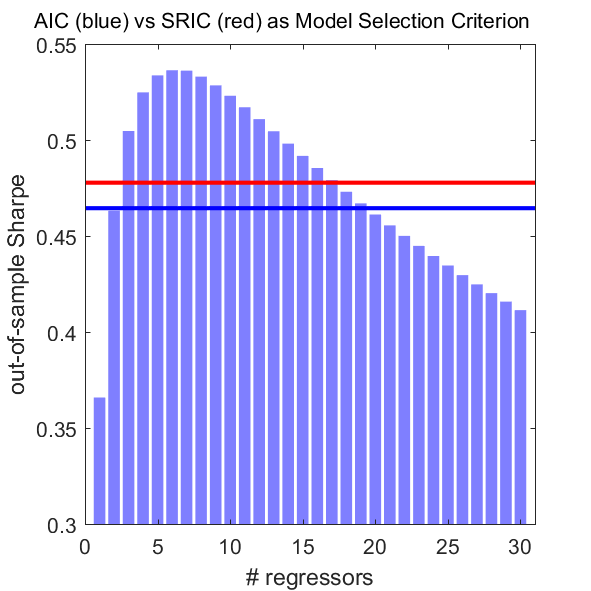}
		\flushleft
	{\textbf{Figure 6 (Regression Simulation):} Average out-of-sample Sharpe ratio by dimension of regression model. The red line denotes the average out-of-sample Sharpe ratio using SRIC, the blue line using AIC. Averages over $100,000$ trials.}
	\label{fig_reg_sim}
\end{figure}

Figure 6 shows the result obtained by averaging over $100,000$ such trials. Note that even though the full $30$-dimensional model is the true model, the out-sample Sharpe ratio peaks at choosing the $6$ dimensional model. This is as the marginal benefit of choosing additional predictors is outweighed by the additional estimation risk. The red line shows the average out-of-sample Sharpe ratio obtained by using SRIC as criterion and the blue line the one obtained by using AIC. Note that both, AIC and SRIC, achieve higher Sharpe ratios than using the true, $30$-dimensional model. For the record, AIC achieves a higher average mean variance utility of $-0.64$ compared to SRIC with $-2.1$.

\subsection{ A Toy Carry Strategy}
\label{applications_reg2}
The previous example was a stylized simulation designed to show the benefit of SRIC.  We now illustrate its applicability in a toy real world application dealing with a simple carry trade strategy. For the avoidance of doubt, this simple example is created for illustration only, and we would not recommend using this strategy. We also do not care about actual out-of-sample performance as for any particular data-set and model, this is quite random anyway. We add this section as a proof of applicability in a real world regression context. 

For this, we use spot and forward prices for $20$ different currencies from January 2000 to October 2015. Each month we rebalance our portfolio consisting of $k+1$ base strategies.
  
The first base strategy is a $12$-month trend following strategy, where the weight is simply the 12-month moving average return (divided by the market's return variance). The other strategies are carry strategies. More precisely: the second base strategy is to set the weight equal to the average last 12-month interest rate differential versus the US-Dollar (divided by the market variance). The interest rate differential is the difference between the (log) spot and the (log) forward price. We call this the \textit{12-month carry strategy}. The third base strategy sets the weight equal to the current interest rate differential (divided by the current market variance). We call this the \textit{lag-$0$ strategy}. The fourth strategy sets the weight equal to the $1$-month lagged interest differential (\textit{lag-$1$ strategy}), the fifth one uses the $2$-month lagged interest rate differentials.  A model with dimension $k$ is now a combination of the $12$-month trend strategy, the $12$-month carry strategy and a number of lagged interest rate differential strategies. Including lagged interest rate differentials allows the model to use changes of interest rates as predictors.

\begin{figure}[H]
	\centering
		\includegraphics[width=18cm]{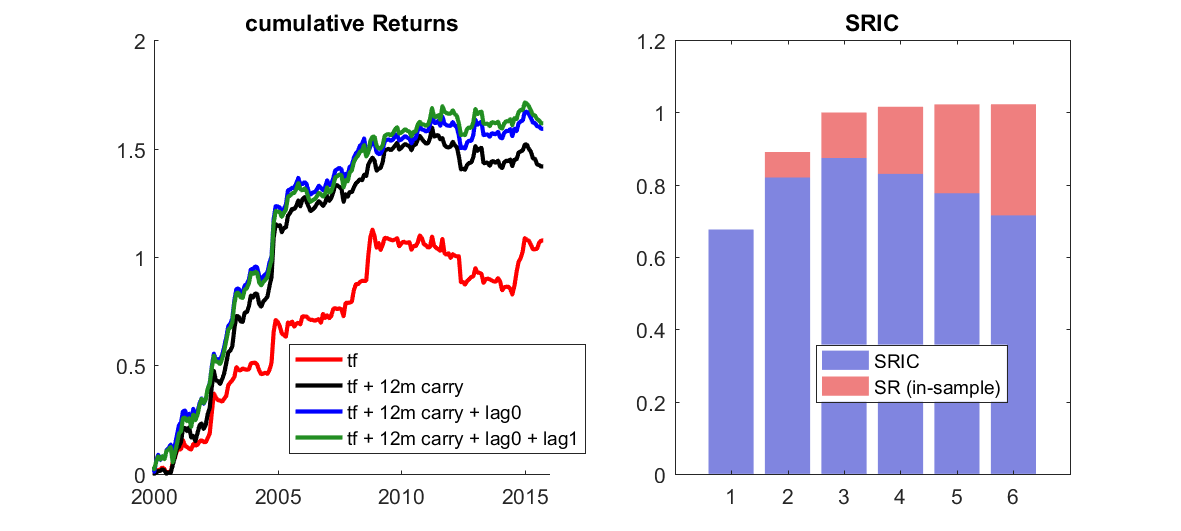}
		\flushleft
	{\textbf{Figure 7:} Left: cumulative returns for different in-sample strategies. Right: in-sample Sharpe ratios and their corrections for noise fit and estimation error (SRIC) for varying numbers of base strategies. }
	\label{fig_example_carry}
\end{figure}

By building the portfolio of base strategies which maximize the in-sample Sharpe ratio, we are essentially performing a regression of the currency returns on the predictions implied by the base strategies (factors) as outlined in Example \ref{example_regression}, equally weighting all currencies.
\begin{eqnarray*}
&& r^j_{t+1} = X \beta  + \eps^j_{t+1} \\
&& \text{with } X= \brae{\operatorname{ma}(r^j_t, 12),\operatorname{ma}(i_t^j-i_t^{\text{\$}}, 12), i_t^j-i_t^{\text{\$}}, i_{t-1}^j-i_{t-1}^{\text{\$}},\ldots} 
\label{eq_exampele_carry}
\end{eqnarray*} 
where $r^j$ is the the excess return of currency $j$,  $i^j$ the interest rate and $ma$ denotes the moving average. 

It is intuitive that including the current interest rate differential adds value but that the additional information provided by including more and more lags diminishes and will at some point be outweighed by the cost in terms of overfitting (noise fit and estimation error). The SRIC derived in the previous section can now serve as a model selection criterion. We illustrate this in Figure 7. The left hand side shows the cumulative returns of the optimal in-sample combination of the $12$-month trend strategy, the $12$-month carry strategy and carry strategies with up to $2$ lags. The right hand side shows the in-sample Sharpe ratio for strategies combining the $12$-month trend strategy with $12$-month carry strategy and up to $4$ lagged carry strategies. The in-sample Sharpe ratio increases with the number of parameters, but after correcting for noise fit and estimation error (SRIC) it reaches a peak after inclusion of lag $0$. That is SRIC would recommend a strategy which combines the trend base strategy with $12$-month carry strategy and the lag-$0$ carry strategy.

\color{black}
\section{Extensions}

\label{section_extensions}
There are several directions into which this article can be extended.

First, this article is restricted to the case in which the returns of the strategy depend linearly on the parameter: $\mu(\theta)=\mu^T\theta$. 
Under suitable regularity assumptions all results are still asymptotically true for non-linear relationships $\mu(\theta)$. Namely, for large times $T$, the estimated parameter will converge to the true parameter. Around the true parameter, the situation will, thanks to the regularity assumptions, be approximately linear and we are back in the case of this article. This, together with control of large (but unlikely) deviations will prove that SRIC is an asymptotically (of order $o(1/T)$) unbiased estimator of the out-of-sample Sharpe ratio in a more general case of non-linear dependencies.


Second, in this paper, we chose to consider a setting with Gaussian noise. However, all results apply asymptotically with \textit{non-Gaussian noise} if the noise has sufficiently bounded higher moments. 

Third, it would be interesting to consider the case in which the \textit{covariance $\Sigma$ is only estimated} and not known. A natural modeling choice would be that $\Sigma_{\text{true}} \Sigma^{-1}_{\text{estimated} }$ is Wishart-distributed. One can then look at the additional loss of going from $\hat \tau(\Sigma_{\text{estimated} })$ to  $\hat \tau(\Sigma_{\text{true} })$ where $\hat \tau (\Sigma)$ denotes the out-of-sample Sharpe ratio of parameter $\hat \theta$ when the true covariance is $\Sigma$.

\section{Conclusion}
\label{section_conclusion}

In this paper, we derived an unbiased estimator of the out-of-sample Sharpe ratio when the in-sample Sharpe ratio is optimized over $k$ parameters. The estimator, which we call SRIC, is a closed form correction of the in-sample Sharpe ratio for noise fit and estimation error. 

We then showed how to apply SRIC as model selection criterion and interpreted it as analogue to the Akaike Information Criterion with the Sharpe ratio rather than log-likelihood (respectively mean-variance utility) as metric for model fit. While model selection via AIC leads to higher out-of-sample mean-variance utility, model selection by SRIC leads to higher out-of-sample Sharpe ratios.  

Several toy applications illustrated its applicability. SRIC is useful whenever estimating the Sharpe ratio net of noise fit and estimation error is of interest, be it as estimator for the out-of-sample Sharpe ratio, or when selecting between models of different dimensions such as factors, asset weights or parameters in prediction models. This makes SRIC particularly applicable within portfolio management. 

	
	\pagebreak
	
	\section{Appendix}

%



\subsection{Proof of Theorem \ref{theorem_linear_case_nf_me} }
\begin{proof}[Proof of Theorem \ref{theorem_linear_case_nf_me}]
 We have to show 
\begin{equation*}
\E\brae{\hat \rho -\hat \tau } = \E\brae{\frac{k}{T \hat \rho} }.
\end{equation*} 
Now by quadratic optimization we have that $\hat \theta =  \Sigma^{-1} \bra{\mu +\nu}$ maximizes the in-sample Sharpe ratio and $\theta^*=  \Sigma^{-1} \mu$ maximizes the out-of-sample Sharpe ratio.  Therefore,
\begin{eqnarray}
\rho(\hat \theta)  &=& \norm{\mu+\nu}_{\Sigma^{-1} } \label{equation_hatrho} \text{ and} \\
\tau( \theta^*) &=& \norm{\mu}_{\Sigma^{-1} }
\end{eqnarray}
with $\norm{x}_{\Sigma^{-1} }=\sqrt{x^T\Sigma^{-1}{x}}$.
 
Plugging $\hat \theta$ into $\tau$ leads to the out-of-sample Sharpe ratio $\hat \tau$
\begin{eqnarray}
\hat \tau  &=& \frac{\mu^T \Sigma^{-1}  \bra{\mu +\nu} }{ \norm{\mu+\nu}_{\Sigma^{-1} }}.
\label{equation_hattau}
\end{eqnarray}
Using (\ref{equation_hatrho}) and (\ref{equation_hattau}), we have to show
\begin{equation*}
\E\brae{  \frac{\nu^T\Sigma^{-1} \bra{\mu+\nu} }{\norm{\mu+\nu}_{\Sigma^{-1} }} } = \E\brae{\frac{k}{T  \norm{\mu+\nu}_{\Sigma^{-1} }} }.
\end{equation*}
Without loss of generality (after a reparametrization) we can assume that $\Sigma=I_{k+1}$, the identity matrix. We then have to show
\begin{equation*}
\E\brae{  \frac{\nu^T\bra{\mu+\nu} }{\norm{\mu+\nu} }}  = \E\brae{\frac{k}{T  \norm{\mu+\nu}} }.
\end{equation*}
We write
 \begin{eqnarray*}
\E\brae{  \frac{\nu^T \bra{\mu+\nu} }{\norm{\mu+\nu}} } 
&=& \sum\limits_{i=1}^{k+1}  \E\brae{  \frac{\nu_i  \bra{\mu_i+\nu_i} }{\norm{\mu+\nu}} } =\sum\limits_{i=1}^{k+1}  A_i.
\end{eqnarray*}
Now integrating out $\nu_i$, 
\begin{eqnarray*}
A_i &=& \E\brae{  \frac{\nu_i  \bra{\mu_i+\nu_i} }{\norm{\mu+\nu}} }\allowdisplaybreaks[1]\\
&=&\E\brae{ \int_{-\infty}^{\infty} \frac{\nu_i  \bra{\mu_i+\nu_i} }{\norm{\mu+\nu}}\sqrt{\frac{T}{2 \pi} } e^{-T \nu_i^2/2} d\nu_i   } \text{ where expectation is over $\nu_j$, $j\not = i$}  \allowdisplaybreaks[1]\\
&=&  \E\brae{ \int_{-\infty}^{\infty}  \frac{1}{T} \sqrt{\frac{T}{2 \pi} }  g(\nu_i) f'(\nu_i) d\nu_i   }
\end{eqnarray*}
with 
\begin{eqnarray*}
f(\nu_i) &=& - e^{-T \nu_i^2/2 } \quad \text{ and }\\
g(\nu_i) &=& \frac{\nu_i + \mu_i}{\norm{\mu+\nu} }
\end{eqnarray*}
so that
\begin{eqnarray*}
f'(\nu_i) &=& \nu_i T  e^{-T \nu_i^2/2 }\\
g'(\nu_i) &=& \frac{1}{\norm{\mu+\nu} } - \bra{\nu_i + \mu_i}\frac{1}{\norm{\mu+\nu}^3}\bra{\nu_i + \mu_i}\\
&=&  \frac{\sum\limits_{j\not =i} \bra{\mu_j +\nu_j}^2 }{\norm{\mu+\nu}^3}.
\end{eqnarray*}
Hence, via integration by parts
\begin{eqnarray*}
A_i &=&  \E\brae{ \int_{-\infty}^{\infty}  \frac{1}{T} \sqrt{\frac{T}{2 \pi} }  g'(\nu_i) \bra{-f(\nu_i) }d\nu_i   } \\
&=& \frac{1}{T} \E\brae{ \int_{-\infty}^{\infty}  \sqrt{\frac{T}{2 \pi} } \frac{\sum\limits_{j\not =i} \bra{\mu_j +\nu_j}^2 }{\norm{\mu+\nu}^3}
e^{-T \nu_i^2/2} d\nu_i   } \\
&=&\frac{1}{T} \E\brae{  \frac{\sum\limits_{j\not =i} \bra{\mu_j +\nu_j}^2 }{\norm{\mu+\nu}^3} } \text{ where expectation now is over $\forall j$ again.}
\end{eqnarray*}
By symmetry, we have
\begin{eqnarray*}
\sum\limits_{i=1}^{k+1}  A_i &=& \frac{1}{T} \sum\limits_{i=1}^{k+1} \E\brae{  \frac{\sum\limits_{j\not =i} \bra{\mu_j +\nu_j}^2 }{\norm{\mu+\nu}^3} } \\
&=& \frac{1}{T} \E\brae{ \frac{k \norm{\mu+\nu}^2}{\norm{\mu+\nu}^3}} \\
&=& \frac{1}{T} \E\brae{ \frac{k }{\norm{\mu+\nu}}}
\end{eqnarray*}
what we had to show.

\end{proof}

\subsection{Proof of Theorem \ref{theorem_nf_ee_split} }

\begin{proof}[Proof of Theorem \ref{theorem_nf_ee_split}] Without loss of generality we can assume $\Sigma=I_{k+1}$, the identity matrix. 
We have that $\hat \rho = \norm{\hat \mu}=\sqrt{\bra{\mu+\nu}^T\bra{\mu+\nu}}$. We now expand $f(\nu)= \sqrt{\bra{\mu+\nu}^T\bra{\mu+\nu}}$ in a Taylor series. 
We have
\begin{eqnarray*}
f(0) &=& \norm{\mu} \\
Df(0)(h) &=& \norm{\mu}^{-1} \bra{\mu^T h}\\
D^2f(0)(h,w) &=&  -\norm{\mu}^{-3} \bra{\mu^Th} \bra{\mu^Tw} +\norm{\mu}^{-1} \bra{w^T h},
\end{eqnarray*}
where we used the assumption $\tau^*>0$ so that $\norm{\mu}>0$.
Hence
\begin{eqnarray*}
f(\nu) &=& \norm{\mu} + \norm{\mu}^{-1} \bra{\mu^T \nu} - \frac{1}{2} \norm{\mu}^{-3} \bra{\mu^T\nu}^2  +\frac{1}{2}\norm{\mu}^{-1} \norm{\nu}^2 +o(\norm{\nu}^2).
\end{eqnarray*}
Taking expectations\footnote{Note that the $o$-term is controlled by its $L^p$-norm and the concentration inequality for Gaussian random variables so that $\E\brae{o(\norm{\nu}^2)}=o(T^{-1})$. See our proof of Theorem \ref{theorem_uncertainties} for a detailed proof of a similar statement. } gives
\begin{eqnarray*}
\E\brae{\hat \rho} &=&\E\brae{f(\nu)}\\
 &=& \norm{\mu} +0 - \frac{1}{2}\frac{1}{\norm{\mu} T} + \frac{1}{2} \frac{k+1}{\norm{\mu} T} + o\bra{T^{-1}}\\
&=& \norm{\mu} + \frac{1}{2}\frac{k}{ \norm{\mu} T} + o\bra{T^{-1}}\\
&=& \tau^*+\frac{1}{2}\frac{k}{\tau^* }\frac{1}{T} + o\bra{T^{-1}}.
\end{eqnarray*}
A first order Taylor development of $\hat\rho^{-1}$ similar to the Taylor development of $f(\nu)=\hat\rho$ above yields that $\E\brae{1/\hat\rho}\to 1/\tau^*$. This allows to substitute $1/\tau^*$ by $\E\brae{1/\hat\rho}$ in the above expression and we obtain

\begin{eqnarray*}
\E\brae{\hat \rho}&=&\tau^*+\frac{1}{2}\frac{k}{\tau^* T} + o\bra{T^{-1}} \\
&=&\tau^*+\E\brae{\frac{1}{2}\frac{k}{\hat \rho T}} + o\bra{T^{-1}}
\end{eqnarray*} and therefore that
\begin{equation*}
\hat \rho-\frac{1}{2}\frac{k}{\hat \rho T}
\end{equation*} is an estimator for $\tau^*$ asymptotically unbiased of order $T^{-1}$. As we already know by the previous theorem that 

\begin{equation*}
\hat \rho-\frac{k}{\hat \rho T}
\end{equation*} is an unbiased estimator for $\hat \tau$ the claimed splitting is proven.
\end{proof}

\subsection{Proof of Theorem \ref{theorem_uncertainties} }
\begin{proof}[Proof of Theorem \ref{theorem_uncertainties}]
Without loss of generality (after reparametrization of $\theta$) let $\Sigma=I_{k+1}$, the identity matrix.
Then by equation (\ref{equation_hatrho}) and (\ref{equation_hattau}) we have
\begin{eqnarray*}
\hat \rho - \hat \tau &=& \frac{\nu^T (\mu + \nu)}{\norm{\mu+\nu}} \leq \norm{\nu}
\end{eqnarray*} with equality when $\mu=0$. 

 If $\mu\not =0$ we have by Tayloring $f(\nu)=\norm{\mu+\nu}^{-1}$ at $\nu=0$ and, again without loss of generality, assuming that $\mu=\norm{\mu} e_1$ is a multiple of the first basis vector:
\begin{eqnarray*}
\hat \rho - \hat \tau &=& \frac{\nu^T (\mu + \nu)}{\norm{\mu+\nu}} \\
&=& \nu^T (\mu + \nu) \bra{\norm{\mu}^{-1} - \norm{\mu}^{-3} \mu^T \nu + o(\norm{\nu})} \\
&=& \norm{\mu}^{-1} \bra{\nu^T (\mu + \nu)  - \norm{\mu}^{-2} \bra{\mu^T \nu}^2  } +R \\
&=& \norm{\mu}^{-1} \bra{\nu^T\nu-\bra{\nu^Te_1}^2 + \norm {\mu}\nu^T e_1 }+R \\
&=& \frac{1}{T \norm{\mu}} \underbrace{T\bra{ \nu^T\nu- \bra{\nu^Te_1}^2 }}_{Z}+ \frac{1}{\sqrt{T}} \underbrace{\sqrt{T} \nu^T e_1}_{N} +R\\
&=& \frac{1}{T \norm{\mu}} Z+ \frac{1}{\sqrt{T}} N +R
\end{eqnarray*}
with $Z$ being $\chi^2(k)$ distributed,  $N$ an independent standard normal distributed random variable and $R= \nu^T (\mu + \nu) o(\norm{\nu}) - \nu^T \nu \norm{\mu}^{-3} \mu^T \nu$. From the above display we see that for all $q\geq1$ the moment $E\brae{\abs{R}^q}$ is uniformly bounded for $T\geq1$.

It remains to show that $\E\brae{T\abs{R}^p}\to 0$ for all $p\geq 1$ when $T\to \infty$. Let $\eps>0$. For a small enough $\delta>0$ holds $\E\brae{T \abs{R}^p; \norm{\nu}\leq \delta }<\eps$.
Now for any $B>1$ we have
\begin{eqnarray*}
\E\brae{T \abs{R}^p} &=&  \E\brae{T \abs{R}^p; \norm{\nu}\leq \delta} +\E\brae{T \abs{R}^p; \norm{\nu} > \delta} \\
&\leq& \eps + T \E\brae{ \abs{R}^p; \{\norm{\nu} > \delta \} \cap \{ R<B\} }+ T \E\brae{ \abs{R}^p; \{\norm{\nu} > \delta \} \cap \{ R\geq B\} }\\
&\leq& \eps+ T B^p \P\brae{\norm{\nu} > \delta}+ T \E\brae{ \abs{R}^{2 p} B^{-p};  \{ R\geq B\} } \\
&\leq&  \eps+ T B^p \P\brae{\norm{\nu} > \delta} + T B^{-p} \E\brae{ \abs{R}^{2 p}}
\end{eqnarray*}

Hence if we choose $B=\P[\norm{\nu} > \delta]^{-\frac{1}{2p}}$,
we have
\begin{eqnarray*}
\E\brae{T \abs{R}^p} &\leq& \eps +  (1+\E\brae{ \abs{R}^{2 p}}) T \P\brae{\norm{\nu} > \delta}^{\frac{1}{2}} \\
&\leq& \eps +  C T e^{-a T} \quad \text{for some $a,C>0$}\\
&\leq& 2 \eps \quad \text{for large $T$}
\end{eqnarray*}
where in the second inequality we used the tail properties of the normal distribution $\nu \sim N(0, \frac{1}{T}I_{k+1})$ and that $\E\brae{ \abs{R}^{2 p}}$ is uniformly bounded in $T$.
\end{proof}

\subsection{Proof of Theorem \ref{theorem_linear_case_squared_nf_me} }
\begin{proof}[Proof of Theorem \ref{theorem_linear_case_squared_nf_me}]

Simple quadratic optimization shows that $\hat \theta =\frac{1}{\gamma} \Sigma^{-1 }\bra{\mu+\nu}$ maximizes in-sample utility $\hat u$, while $\theta^* =\frac{1}{\gamma} \Sigma^{-1 }\mu$ maximizes out-of-sample utility $u$.  With this $\hat u (\hat \theta)=\frac{1}{\gamma} \hat \rho^2$ as well $u ( \theta^*)=\frac{1}{\gamma} {\tau^*}^2$.

Now it is straightforward to see that for mean-variance utility

\begin{eqnarray*}
\E\brae{\N_{MV}} &=& \E\brae{\hat u(\hat \theta)}- \E\brae{\hat u(\theta^* \phantom{\hat \theta}  )}\\
&=& \E\brae{ \frac{1}{\gamma} \rho(\hat \theta)^2}- \frac{1}{\gamma} \tau(\theta^*)^2 +0\\
&=& \frac{1}{\gamma} \E\brae{ \norm{\mu+\nu}^2_{\Sigma^{-1} } -\norm{\mu}^2_{\Sigma^{-1} } }\\
&=& \frac{1}{\gamma} \E\brae{ 2\mu \Sigma^{-1} \nu + \nu \Sigma^{-1}\nu }\\
&=& \frac{k+1}{\gamma T}.
\end{eqnarray*}

Similarly 
\begin{eqnarray*}
\E\brae{\N_{MV} + \EE_{MV}+\U_{MV}} &=& \E\brae{\hat u(\hat \theta)}- \E\brae{ u(\hat\theta)}\\
&=& \frac{1}{\gamma} \E\brae{ \norm{\mu+\nu}^2_{\Sigma^{-1} } - \bra{2 \mu^T \Sigma^{-1} \bra{\mu +\nu}- \norm{\mu+\nu}^2_{\Sigma^{-1} } }}\\
&=& \frac{2 (k+1)}{\gamma T}
\end{eqnarray*}
The simplicity of the proof is due to the beauty of the mean-variance utility and squared entities: its geometry is linear. 

\end{proof}


\subsection{Proof of Theorem \ref{theorem_loglikelihood}}

\begin{proof}[Proof of Theorem \ref{theorem_loglikelihood}]

We start by describing a set-up in which the returns $s^{\theta}_t = r_t \theta$ are derived within a setting of market predictions. Though not necessary, we do it in a continuous time setting here, as this will technically more elegant and less tedious (e.g. when it comes to the Girsanov Theorem). Note that it also comprises the discrete time case. 

Roughly speaking, the assumption that we need is that the system returns $s^{\theta}$ are obtained by betting a weight $w_t$ equal to a prediction on markets with return $r_t$. We do not need to assume, however, that the predictions are actually observed. 

For that let $p_t$ be a process of cumulative excess returns on $m$ markets given by a solution to
\begin{equation}
 \frac{dp_t}{p_t} = Y_t dt + S_t dW_t 
\label{eq_market_dynamics}
\end{equation} with a $m$-dimensional Brownian motion $W$, a deterministic process $\Y=(Y_t)_{t\in \R}$ and $S_t\in \R^{m,m}$ deterministic and invertible.  Here,  $ \Y$ is the (unknown) predictable component of the excess returns. Without loss of generality assume that $S_t= I_m$, the identity matrix. Otherwise we rotate and change the leverage of the markets $p_t$. 

%

Let \begin{equation*}
\X: [0,T]\times \Theta \to \R^m \quad\text{ with }   \X(\theta)\in L^2([0,T]) \quad \forall \theta \in \Theta
\nonumber
\end{equation*} be a parametrized deterministic process $ \X = ( X_t)_{t\geq 0}$, the (parametrized) predictions. Because in this article we restrict ourselves to linear dependencies, we have $ X_t(\theta) =X_t \theta$ with some abuse of notation.  

The assumption is now that the system returns are given by betting with weight $w_t= X_t/p_t$ on the markets $p_t$ (note that these weights maximize the Sharpe ratio if the expected return is $X_t$). Doing so leads to the cash-flow
\begin{eqnarray*}
d s_t^{\theta} &= & w_t^T d p_t \\
&=& X_t(\theta)^T Y_t  dt +  X_t(\theta)^T dW_t.
 \nonumber
\end{eqnarray*} 
We derive $\Sigma(\theta, \theta')$,  the (annualized) quadratic covariation of $s^\theta$ as
\begin{eqnarray}
\Sigma(\theta, \theta')&=& \frac{1}{T} \int_0^T  X_t^{T}(\theta) X_t(\theta') dt \label{equation_sigma}
\end{eqnarray}  

and $\hat \mu$, the (annualized) realized  returns as
\begin{eqnarray}
\hat \mu(\theta) &=& \frac{1}{T} \int_0^T X_t(\theta)^T Y_t dt +  \frac{1}{T} \int_0^T X_t(\theta)^T dW_t \label{equation_mu}\\
&=& \mu(\theta) +\nu(\theta)
\end{eqnarray}
where $\nu(\theta)$ is a random variable with covariance $\frac{1}{T} \Sigma$. This is exactly the setup\footnote{Using the linearity $X_t(\theta)=X_t \theta $ yields with a slight abuse of notation $\Sigma(\theta, \theta')=\theta^T \Sigma \theta$,
$\mu(\theta)=\mu^T \theta $ and $\nu(\theta)=\nu^T \theta$, with $\Sigma =\frac{1}{T} \int_0^T  X_t^{T} X_t dt $, 
$\mu=\frac{1}{T} \int_0^T X_t^T Y_t dt $ and $\nu=\frac{1}{T} \int_0^T X_t^T dW_t$.} of section \ref{section_setup}. That is $\theta$ parametrizes linearly investment strategies with estimated return $\hat \mu$ and covariance $\Sigma$. However, now, a parameter $\theta$ is related to a prediction $X_t(\theta)$ so that we can derive a (log) likelihood.

We are now able to derive the log-likelihood of \textit{prediction} $\theta$. 

By (\ref{eq_market_dynamics}), the dynamics of the market prices are given by (note that without loss of generality $S_t=I_m$)
	\begin{equation*}
		\frac{dp_t}{p_t} = Y_t dt +  dW_t 
		\nonumber
	\end{equation*}

with unknown $Y_t$. Now $\theta$ parametrizes different models $\P^{\theta}$ for $p_t$

	\begin{equation}
		\frac{dp_t}{p_t} = X_t(\theta) dt +  dW^{\theta}_t 
		\label{eq_market_model_theta}
	\end{equation}
	with $W^{\theta}$ a Brownian motion under $\P^{\theta}$.
The question is: What is the (log)likelihood of the realized market price process $p$ under the model $\theta$? The answer is given by the help of the Girsanov Theorem. Due to continuous time, each likelihood (density) is zero. So the only meaningful definition of the log likelihood function is as a relative density with respect to a reference probability. 

A natural reference measure is the probability measure in which there is zero predictability, that is a geometric Brownian motion. For this purpose let
	\begin{equation}
		\frac{dp_t}{p_t} = dW^0_t 
	\end{equation}
be a geometric Brownian motion under the probability distribution $\P^0$. 

Let
$Z_T = T \hat \mu(\theta) - \frac{1}{2} T \Sigma(\theta, \theta)$

Then
\begin{eqnarray*}
Z_T &=& 
\int_0^T X_t(\theta)^T Y_t dt + \int_0^T X_t(\theta)^T  dW^{\theta}_t  \\
&&-\frac{1}{2} \int_0^T X_t(\theta)^T X_t(\theta) d t \\
&=& \int_0^T X_t(\theta)^T  dW^0_t -\frac{1}{2} \int_0^T X_t(\theta)^T X_t(\theta) d t. 
\end{eqnarray*}
Hence by the (multivariate) Girsanov Theorem (see any book on stochastic analysis, e.g. \cite{kallenberg2002foundations} or  \cite{oksendal2003stochastic}), $ \tilde W_t =W^0_t-\operatorname{diag}([Z,W^0]_t)= W^0_t - \int_0^T X_t(\theta) dt$ is a Brownian motion under $\Q^{\theta} =\exp(Z_
{T}) \P^{0}$. In particular under $\Q^{\theta}$:
	\begin{equation*}
		\frac{dp_t}{p_t} =  d W^0_t = d \tilde W_t + X_t(\theta) dt \nonumber
	\end{equation*}
	where is $\tilde W_t$ a Brownian motion. Hence $\Q^{\theta} \stackrel{\D}{=}\P^{\theta}$ is the probability measure we are looking for and $\frac{d\P^{\theta}}{d\P^{0}} =\frac{d\Q^{\theta}}{d\P^{0}} =e^{Z_T} $. Taking logs finishes the proof.
\end{proof}

\subsection{Relation to Siegel and Woodgate: Performance of Portfolios Optimized with Estimation Error}
There is an interesting though at first sight not obvious connection to \cite{siegel2007performance}. We show how to use their results as estimator for the out-of-sample Sharpe ratio in our setting as described in Section~\ref{section_setup}. We then show that even though the authors prove that their estimators of mean and variance are asymptotically unbiased of order $\frac{1}{T^2}$ the resulting estimator for the Sharpe ratio is biased of order $\frac{1}{T }$, the same order of magnitude as without the adjustment. 

To see this, note first that we can do without the volatility adjustment in \cite{siegel2007performance} as increasing the sampling frequency and annualizing mean, variance and Sharpe ratio lets the variance adjustment converge to zero and the mean adjustment intact.
 
Second, assume that there is one riskless asset so that both $\mu_*$ and $\sigma^2_*$, that is mean and variance of the minimum variance portfolio in the notation of Siegel and Woodgate, equal zero. To be precise a riskless asset would make $\Sigma$ degenerate, but we can think of a sequence in which one asset converges to a riskless asset. 

In the limit (infinite sampling frequency, one riskless asset) we get, for the bias-adjusted estimates of mean and variance, by using formula (8) in \cite{siegel2007performance}: 
\begin{eqnarray}
\hat \mu_{adj} &=& \mu_0 - \frac{n-3}{T} \hat B_{2,2} \mu_0\nonumber\\
\hat \sigma_{adj} &=& \sigma_0\nonumber
\end{eqnarray}
where $\hat B_{2,2}=\frac{1}{\hat \rho^2}$ (as $\mu_*=0$), $\mu_0$ is the target mean of the mean-variance portfolio, $\sigma_0$ its in-sample volatility and $n$ is the number of assets in the Siegel--Woodgate framework. 

Hence, we obtain the estimated out-of-sample Sharpe as
\begin{eqnarray}
\frac{\hat \mu_{adj}}{\hat \sigma_{adj}} &=& \hat\rho - \frac{n-3}{T \hat \rho}
\label{eq_siegel2sric}
\end{eqnarray}
where $\mu_0/\sigma_0=\hat \rho$ is the maximum in-sample Sharpe ratio. Note that due to the riskless asset the efficient frontier is a line and the Sharpe ratio the same for all target returns $\mu_0$ . 

Formula (\ref{eq_siegel2sric}) is similar to SRIC, but with an adjustment of $n-3$ rather than $k=n-2$ which would be the adjustment for an unbiased estimator. To see this, note that $w \in \R^{n}$ has $n-1$ degrees of freedom and we have to account for the additional riskless asset (which does not influence the Sharpe ratio and cannot be counted), bringing the degrees of freedom to $n-2$. Hence the bias correction in \cite{siegel2007performance} lacks one degree of freedom when it comes to the Sharpe ratio.

	\bibliographystyle{apa}
	\bibliography{bibliography}

\end{document}